\documentclass[notitlepage,a4,11pt]{article}%,twocolumn]{revtex4}

\usepackage{graphicx}

\usepackage{amsmath}%
\usepackage{amsfonts}%
\usepackage{amssymb}%
\usepackage{mathrsfs}
\usepackage{amsthm}

\usepackage{authblk}
\newcommand{\mc}{\mathcal}
\newcommand{\cp}{\times}

\newcommand{\di}{\nabla\cdot}
\newcommand{\cu}{\nabla\times} 
\newcommand{\JI}[1]{\bol{#1}\cdot\cu{\bol{#1}}}
\newcommand{\BN}[1]{\bol{#1}\cp\left(\cu{\bol{#1}}\right)}

\newcommand{\bol}{\boldsymbol}

\newcommand{\abs}[1]{\left\lvert{#1}\right\rvert}
\newcommand{\w}{\wedge}
\newcommand{\lr}[1]{\left({#1}\right)}

\newcommand{\mf}{\mathfrak}
\newcommand{\p}{\partial}

\newcommand{\ti}[1]{\textit{#1}}
\newcommand{\tb}[1]{\textbf{#1}}
\newcommand{\ov}[1]{\mkern 1.5mu\overline{\mkern-1.5mu#1\mkern-1.5mu}\mkern 1.5mu}
\newcommand{\B}[3]{\cos{#3}\,\nabla{#2}+\sin{#3}\, \nabla{#1}}
\newtheorem{mydef}{\textit{Def}}%[section]
\newtheorem{remark}{\textit{Remark}}%[section]
\newtheorem{theorem}{\textit{Theorem}}%[section]
\newtheorem{proposition}{\textit{Proposition}}%[section]
\newtheorem{example}{\textit{Example}}%[section]
\begin{document}
\title{Existence of Ideal Magnetofluid Equilibria\\ Without Continuous Euclidean Symmetries}
%Magnetofluidostatic Fields\\ Without Euclidean Symmetries}
\author[1]{N. Sato} %\author[1]{Z. Yoshida}
\affil[1]{Research Institute for Mathematical Sciences, \protect\\ Kyoto University, Kyoto 606-8502, Japan \protect\\ Email: sato@kurims.kyoto-u.ac.jp}
\date{\today}
\setcounter{Maxaffil}{0}
\renewcommand\Affilfont{\itshape\small}

\maketitle

\begin{abstract}
We study the existence of steady solutions 
of ideal magnetofluid systems (ideal MHD and ideal Euler equations) 
without continuous Euclidean symmetries. 
It is shown that all nontrivial magnetofluidostatic solutions 
are locally symmetric, 
although the symmetry is not necessarily an Euclidean isometry.
Furthermore, magnetofluidostatic equations admit
both force-free (Beltrami type) and non-force-free (with finite pressure gradients) solutions that do not exhibit invariance under translations, rotations, or their combination.
Examples of smooth solutions without continuous Euclidean symmetries in bounded domains are given.
Finally, the existence of square integrable solutions of the tangential 
boundary value problem without continuous Euclidean symmetries is proved. 
\end{abstract}

\section{Introduction}
Ideal magnetofluid equilibria are described by magnetofluidostatic fields. 
A Magnetofluidostatic field is a solution $\bol{w}\in C^1\lr{\Omega}$ of the system of first order partial differential equations
\begin{equation}
\BN{w}=\nabla\chi,~~~~\di\bol{w}=0~~~~{\rm in}~~\Omega.\label{MFF}
\end{equation}
Here, $\Omega\subset\mathbb{R}^3$ denotes a bounded domain, and $\chi\in C^1\lr{\Omega}$ a function. For the purpose of the present paper, 
the function $\chi$ is not given, and we shall say that $\bol{w}$ 
solves \eqref{MFF} provided that the system is satisfied for some
appropriate choice of $\chi$. 
In the second part of the paper, we will consider system \eqref{MFF} with 
tangential boundary conditions
\begin{equation}
\bol{w}\cdot\bol{n}=0~~~~{\rm on}~~\p\Omega.\label{BC}
\end{equation}  
Here, $\bol{n}$ denotes the unit outward normal to the boundary $\p\Omega$, 
which is assumed smooth.

In the context of magnetohydrodynamics, $\bol{w}$ represents the magnetic field,
$\cu\bol{w}$ the current density, and $-\chi$ the pressure field. 
System \eqref{MFF} then expresses force balance between magnetic force $-\BN{w}$ and pressure force $-\nabla\chi$ (see for example \cite{Kruskal1958}). 
On the other hand, if the constant density ideal Euler equations of fluid dynamics are considered, 
$\bol{w}$ represents the incompressible flow velocity, $\cu\bol{w}$ 
the vorticity, and $\chi$ the so called Bernoulli head (the sum of squared
velocity and pressure). System \eqref{MFF} then describes
a steady incompressible Euler flow \cite{Moffatt2014}. 
In the following, we will refer to $\chi$ as the pressure field.

Regular solutions of system \eqref{MFF} are highly
desirable in experimental applications involving the
physical containment of fluid and plasma systems. 
At present, the design of magnetic confinement devices is an active
area of research driven by the development of nuclear fusion reactors \cite{Boozer2004,Hudson2018}. Here, the topology and symmetry properties of the solutions  
are a key factor for the stable containment of the burning plasma.
The topology of the domain is particularly important because
it may determine whether solutions of \eqref{MFF} exist or not. 
For example, if the domain is spherical and $\chi$ is a function of the sphere radius, from the hairy ball theorem there is no continuous non-vanishing vector field $\bol{w}$ always tangent to the level sets of constant radius, and the first equation of system \eqref{MFF} cannot be satisfied. 

From a mathematical standpoint, system \eqref{MFF} falls in the category
of elliptic-hyperbolic partial differential equations (see for example \cite{Yoshida1990}). 
The classification of a system of first order partial differential equations for $m$ unknowns $u=\left\{u^1,...,u^m\right\}$ in the form 
\begin{equation}
Lu=\sum_{i=1}^{n}A^{i}u_{i}+B=0, \label{L}
\end{equation}
where $A^i\lr{\bol{x},u}$, $i=1,...,n$, %$u_i=\p u/\p x^i$, 
and $B\lr{\bol{x},u}$ are $m\cp m$ matrices
depending on the variables $\bol{x}=\lr{x^1,...,x^n}$ and the unknowns $u$, 
is determined by the characteristic equation 
\begin{equation}
Q\lr{d\phi}=\det\lr{\sum_{i=1}^n A^i \phi_{i}}=0.\label{Q}
\end{equation}
In this notation, a lower index denotes derivation, e.g. $u_i=\p u/\p x^i$.
Any solution $\phi$ of the characteristic equation \eqref{Q} with a non-vanishing
gradient defines a characteristic surface of the operator $L$. 
For the case of system \eqref{MFF} we have $u=\left\{w^1,w^2,w^3,-\chi\right\}$ 
and
\begin{equation}
\sum_{i=1}^{3}A^i \phi_i=
\begin{bmatrix}
w^2\phi_2+w^3\phi_3& -w^2\phi_1& -w^3\phi_1& -\phi_1\\
-w^1 \phi_2& w^1\phi_1 +w^3 \phi_3& -w^3\phi_2& -\phi_2\\
-w^1 \phi_3& -w^2 \phi_3& w^1 \phi_1+w^2 \phi_2& -\phi_3\\
\phi_1& \phi_2& \phi_3& 0
\end{bmatrix}.
%\begin{split}
%{
%A^1={\begin{bmatrix}
%0 & -w_2 & -w_3 & -1\\ 
%0&w_1&0&0\\
%0&0&w_1&0\\
%1&0&0&0
%\end{bmatrix}},~
%A^2={\begin{bmatrix}
%w_2 & 0 & 0 & 0\\ 
%-w_1&0&-w_3&-1\\
%0&0&w_2&0\\
%0&1&0&0
%\end{bmatrix}},~
%A^3={\begin{bmatrix}
%w_3 & 0 & 0 & 0\\ 
%0&w_3&0&0\\
%-w_1&-w_2&0&-1\\
%0&0&1&0
%\end{bmatrix}}.}
%\end{split}
\end{equation}
The resulting characteristic equation is
\begin{equation}
Q\lr{d\phi}=\lr{\nabla\phi}^2\lr{\bol{w}\cdot\nabla\phi}^2=0.
\end{equation}
Hence, system \eqref{MFF} is twice elliptic 
(the first quadratic term gives the trivial characteristic surface $\nabla\phi=\bol{0}$) and twice hyperbolic (the second quadratic term gives a nontrivial
characteristic surface with double multiplicity $\bol{w}\cdot\nabla\phi=0$).
This twofold nature is the reason why, so far, approaches based on standard analysis have been unsuccessful in answering the general problem of existence of solutions \cite{LoSurdo}. Currently, a rigorous mathematical theory of system \eqref{MFF} is not available.

Solutions of \eqref{MFF} can be divided in two groups: 
Beltrami-type solutions, physically corresponding to a constant pressure field $\nabla\chi=\bol{0}$, and those with a non-vanishing pressure force $\nabla\chi\neq\bol{0}$. % in some region $U\subseteq\Omega$. 
Among Beltrami-type solutions, there are two subclasses: 
trivial Beltrami fields, satisfying $\cu\bol{w}=\bol{0}$, %in the domain $\Omega$, 
and Beltrami fields with a non-vanishing curl $\cu\bol{w}\neq\bol{0}$. %in some region of the domain $\Omega$. 
We shall say that a Beltrami field is nontrivial in a given domain if 
$\bol{w}\neq\bol{0}$ and $\cu\bol{w}\neq\bol{0}$ there (this is equivalent to demanding that the helicity density is non-zero, $\JI{w}\neq 0$, because the curl of a Beltrami field is aligned to the field itself). 
Similarly, it is convenient to classify solutions with a non-zero pressure force in two groups: those solutions that possess at least one continuous Euclidean symmetry (see \cite{Ratcliffe} for the definition of Euclidean symmetry), and those that do not (we will show in section 4 that such solutions exist). The distinction between discrete Euclidean symmetries (reflections) and continuous Euclidean symmetries (translations and rotations) will be explained in section 3.

Of the four subclasses, the first three admit a systematic mathematical treatment. Trivial Beltrami field solutions of system \eqref{MFF} 
together with boundary conditions \eqref{BC} can be obtained 
by solving the Neumann boundary value problem for Laplace's equation:
\begin{equation}
\Delta f=0~~~~{\rm in}~~\Omega,~~~~\nabla f\cdot\bol{n}=0~~~~{\rm on}~~\partial\Omega.
\end{equation}   
Then, the trivial solution is $\bol{w}=\nabla f$. 

The existence of nontrivial Beltrami field solutions in the form
\begin{equation}
\cu\bol{w}=\hat{h}\,\bol{w},\label{BF}
\end{equation}
where the proportionality coefficient $\hat{h}$ is, in general, a function, 
has been proven in \cite{YRot}. In particular, theorem 2 of \cite{YRot}
shows that strong solutions of \eqref{BF} satisfying the boundary conditions
\eqref{BC} exist for any constant proportionality coefficient $\hat{h}\in\mathbb{C}$ in the complex numbers if the domain $\Omega$ is multiply connected; 
the admissible (constant) values of $\hat{h}$ become discrete if the
domain is simply connected. It should be noted that the result of \cite{YRot}
applies to smoothly bounded domains $\Omega$ of arbitrary shape. In particular, 
the corresponding solutions will not exhibit, in general, Euclidean symmetries (to confirm this fact, we will construct explicit examples of nontrivial Beltrami field solutions without continuous Euclidean symmetries in section 3). 
An open mathematical problem remains for the case in which $\hat{h}$ is allowed to be a function rather than a constant, namely the identification of the class of functions $\hat{h}$ that admit a corresponding nontrivial Beltrami field solution of system \eqref{MFF}. A result in this direction was obtained by
\cite{Enciso}, which showed that for $\hat{h}$ in an open and dense subset of $C^k\lr{\Omega}$, $k\geq 7$, the only nontrivial Beltrami field solution of system \eqref{MFF} is $\bol{w}=\bol{0}$.  
Nevertheless, this result does not prevent the existence of Beltrami fields with non-constant $\hat{h}$; indeed, it has been shown in \cite{SatoB1,SatoB2} that an infinite number of such solutions exist, and a systematic method to construct them was derived. This method relies on a     
local Clebsch-like parametrization of nontrivial Beltrami field solutions  stemming from the Lie-Darboux theorem of differential geometry \cite{DeLeon89,Arnold89,Salmon17}.
In particular, given a neighborhood $U$ of a point $\bol{x}\in\Omega$, any nontrivial Beltrami field admits the local representation
\begin{equation}
\bol{w}=\cos x^3\,\nabla x^2+\sin x^3\,\nabla x^1~~~~{\rm in}~~U\label{LBF}
\end{equation}
where $\lr{x^1,x^2,x^3}\in C^{\infty}\lr{U}$ is a smooth
curvilinear coordinate system whose contravariant metric tensor $g^{ij}=\nabla x^i\cdot\nabla x^j$ satisfies:
\begin{subequations}\label{MT}
\begin{align}
\cos x^3 \sin x^3\lr{g^{22}-g^{11}}&=g^{12}\lr{\cos^2 x^3-\sin^2 x^3},\label{MT1}\\
\sin x^3 g^{13}+\cos x^3 g^{23}&=0,\label{MT2}\\
\cos x^3\lr{g^{31}+\Delta x^2}+\sin x^3 \lr{\Delta x^1-g^{32}}&=0.\label{MT3}
\end{align}
\end{subequations} 
Here, equation \eqref{MT3} ensures that the solution is solenoidal, i.e. $\di\bol{w}=0$. 
This condition can be discarded
when considering general Beltrami fields. However, in this paper we will be concerned only with solenoidal solutions. 
System \eqref{MFF} is thus converted into a set
of geometric conditions for the metric tensor. 
A nontrivial Beltrami field solution can then be obtained by finding a coordinate system satisfying system \eqref{MT}. 

The treatment of the third class of solutions requires 
the notion of symmetry. In particular, it is known that an intimate relationship exists between the symmetry properties of the solutions, and their existence. According to a conjecture due to H. Grad, only `highly symmetric' solutions of system \eqref{MFF} should be expected \cite{Grad3}. 
While Grad's idea of symmetry involved considerations on regularity, stability, and boundary conditions, in the context of moder plasma physics it is customary to consider a solution as symmetric if the components $w^i$ of $\bol{w}$, the pressure field $\chi$, and the components of the covariant metric tensor $g_{ij}=\p_{i}\cdot\p_{j}$ admit an ignorable coordinate $x^3$, i.e. their derivative with respect to $x^3$ is always zero. Here, $\p_i$ denotes the tangent vector in the $x^i$ direction.
As we will see in section 3, only the subset of Euclidean isometries 
corresponding to the 
special Euclidean group $SE\lr{3}$ (continuous transformations that preserve distance between points and orientation in Euclidean space, namely translations, rotations, and their combination) are compatible with an ignorable coordinate for the metric tensor. 
%In the following, we shall refer to solutions that do not posses an Euclidean symmetry as asymmetric. 
Under these circumstances, it is always possible to remove the hyperbolic part of system \eqref{MFF}, and reduce it to a single nonlinear elliptic second-order partial differential equation, the Grad-Shafranov equation \cite{Grad1,Grad2,Eden1,Eden2} for the flux function $\Theta$:
\begin{equation}
\Delta\Theta-\nabla\Theta\cdot\nabla\log g_{33}-g_{33}\frac{d\chi}{d\Theta}+\frac{1}{2}\frac{dw_3^2}{d\Theta}-g_{33}w_3\nabla\lr{\frac{\p_3\cp\nabla x^3}{g_{33}}}=0.\label{GSE}
\end{equation}
Upon prescribing the functions $w_{3}$ and $\chi$ as 
functions of $\Theta$, equation
\eqref{GSE} can be solved with corresponding solution of system \eqref{MFF} given by
\begin{equation}
\bol{w}=\nabla\Theta\cp\nabla x^3+w^3\,\p_3.
\end{equation} 
Here, $w^3=g^{3j}w_j$. 
Families of analytic solutions of the Grad-Shafranov equation are known, and
can be constructed by using the symmetry group of the equation \cite{White2009}. 
We also remark that the Grad-Shafranov equation does not apply to Euclidean symmetries involving
reflections, since they fall in the category of discrete symmetries.

Little is known about the fourth class, i.e. solutions of \eqref{MFF} with 
non-vanishing pressure gradients and without continuous Euclidean symmetries.
In \cite{Bruno}, Bruno and Lawrence showed that solutions in toroidal domains without symmetry can be constructed if one is willing to postulate a stepped pressure profile. Then, the solution is of Beltrami-type in regions where the pressure gradient vanishes. These regions are divided by nested flux surfaces in correspondence of the pressure `jumps'. Across such surfaces total pressure balance is satisfied (the sum of mechanical and magnetic pressure is constant across the interfaces). 
Numerical solutions of \eqref{MFF} in terms of stepped pressure equilibria 
are routinely employed in the design and study of fusion reactors (stellarators) whose shape do not exhibit continuous Euclidean symmetries \cite{Hudson2012,Hudson2017}. 
At present, the question remains open whether solutions without continuous Euclidean symmetries and with non-vanishing 
pressure gradients exist beyond the stepped pressure case, which represents
the `minimal' departure from Beltrami-type solutions since the pressure force
is non-vanishing only over a set of measure zero. 
The work of Weitzner \cite{Weitzner} supports the possibility that
such solutions do exist. Indeed, in \cite{Weitzner}, a generalized Grad-Shafranov equation accounting for solutions without continuous Euclidean symmetries is derived (although the equation holds locally, it is not elliptic, and it is not the result of a reduction by symmetry), and potential solutions without continuous Euclidean symmetries are expressed in the form of an expansion in a small parameter measuring the departure from a symmetric torus. 

The purpose of the present paper is to establish the existence of solutions
of system \eqref{MFF} that do not possess continuous Euclidean symmetries, i.e. 
solutions that are not invariant under translations, rotations, or their combination. In particular, we will provide examples of smooth
solutions of system \eqref{MFF} with both vanishing and non-vanishing pressure gradients and without continuous Euclidean symmetries in bounded domains, and show that square integrable solutions of the boundary value problem \eqref{MFF}, \eqref{BC} with non-vanishing pressure gradients and  without continuous Euclidean symmetries exist. 
 
The paper is organized as follows. 
In section 2 we review the notion of symmetry and show that
all nontrivial solutions of system \eqref{MFF} are locally symmetric, 
although the symmetry is not necessarily an Euclidean isometry. 
In section 3 we first introduce Euclidean symmetries.
Then, we discuss the class of proportionality coefficients $\hat{h}$ which admit a corresponding Beltrami field, and 
construct nontrivial smooth Beltrami field solutions of system
\eqref{MFF} without continuous Euclidean symmetries in bounded domains.
By using an appropriate Clebsch parametrization, in section 4 we construct classes of smooth solutions of system \eqref{MFF} 
with
non-vanishing pressure gradients and 
without continuous Euclidean symmetries in bounded domains. Then, we prove the existence of square integrable solutions of system \eqref{MFF} with boundary conditions \eqref{BC}, with non-vanishing pressure gradients, and without continuous Euclidean symmetries. This is achieved by separating the bounded domain into a central region and
a peripheral region. Then, solutions with non-vanishing pressure gradients and without continuous Euclidean symmetries in the central region are combined with Beltrami-type solutions in the peripheral region so that boundary conditions are satisfied. 
Finally, in section 5 we discuss the properties of the commutator of curl operator and Lie derivative %with respect to continuous Euclidean symmetries 
when applied to Beltrami fields.
These properties give rise to symmetries of the Beltrami equation that can be used to construct new solutions 
without continuous Euclidean symmetries from known ones.   
Conclusions are drawn in section 6. 

\section{Symmetric solutions}

Let $\Omega\subset\mathbb{R}^3$ denote a smoothly bounded domain
with boundary $\p\Omega$. 
For the purpose of the present paper, we shall refer to a property as `local'
in the sense that it holds in a neighborhood $U$ of a chosen point of interest $\bol{x}\in\Omega$.   
Let $T$ be a $p$ times covariant and $q$ times
contravariant tensor in $\Omega$. In the following, we adopt the definition of symmetry for the tensor $T$ below:
\begin{mydef}
The tensor $T$ is symmetric in $\Omega$ with respect to a vector field $\bol{\xi}\in T\Omega$ if 
\begin{equation}
\mf{L}_{\bol{\xi}}T=0~~~~{\rm in}~~\Omega. \label{Sym}
\end{equation}
Here, $T\Omega$ denotes the tangent space of $\Omega$ and $\mf{L}$ the Lie-derivative. 
\end{mydef}

\noindent Below, we refer to the vector field $\bol{\xi}$ as a symmetry of the tensor $T$. 
When $T$ is a vector field, i.e. $T=\bol{v}$ with $\bol{v}\in T\Omega$, equation
\eqref{Sym} becomes
\begin{equation}
\mf{L}_{\bol{\xi}}\bol{v}=\xi^i\p_i v^j\p_j-v^i\p_i \xi^j \p_j=\lr{\bol{\xi}\cdot\nabla}\bol{v}-\lr{\bol{v}\cdot\nabla}\bol{\xi}=\bol{0}~~~~{\rm in}~~\Omega.\label{Sym2}
\end{equation}
Next, suppose that both the direction of symmetry $\bol{\xi}$ and
the vector field $\bol{v}$ are solenoidal, i.e. $\nabla\cdot\bol{\xi}=\nabla\cdot\bol{v}=0$ in $\Omega$.
Then, from the vector identity
\begin{equation}
\cu\lr{\bol{v}\cp\bol{\xi}}=\lr{\nabla\cdot\bol{\xi}}\bol{v}-\lr{\nabla\cdot\bol{v}}\bol{\xi}
+\lr{\bol{\xi}\cdot\nabla}\bol{v}-\lr{\bol{v}\cdot\nabla}\bol{\xi}, 
\end{equation}
equation \eqref{Sym2} reduces to
\begin{equation}
\mf{L}_{\bol{\xi}}\bol{v}=\cu\lr{\bol{v}\cp\bol{\xi}}=\bol{0}~~~~{\rm in}~~\Omega.\label{Sym3}
\end{equation}
We have the following:

\begin{proposition}\label{prop1}
Assume $\bol{w}\neq\bol{0}$ and $\cu\bol{w}\neq \bol{0}$ in $\Omega$. Then, all smooth solutions $\bol{w}\in C^{\infty}\lr{\Omega}$ of \eqref{MFF} are locally symmetric. Furthermore, the local symmetry $\bol{\xi}$ can always be chosen as solenoidal, i.e. $\di\bol{\xi}=0$. 
\end{proposition}

\begin{proof}
There are two cases. First, suppose that $\nabla\chi=\bol{0}$ in $\Omega$. 
Then, $\bol{w}$ is a Beltrami-type solution such that $\cu\bol{w}=\hat{h}\,\bol{w}$ for some proportionality coefficient $\hat{h}$. 
According to theorem 1 of \cite{SatoB1}, for every $\bol{x}\in\Omega$ there exists
a neighborhood $U\subset\Omega$ of $\bol{x}$ and smooth curvilinear coordinates $\lr{\ell,\psi,\theta}\in C^{\infty}\lr{U}$ such that
\begin{equation}
\bol{w}=\cos\theta\,\nabla\psi+\sin\theta\,\nabla\ell=w^2\lr{\cos\theta\,\p_{\psi}+\sin\theta\,\p_{\ell}}~~~~{\rm in}~~U.\label{w}
\end{equation}
In this notation $w^{2}=\sqrt{\bol{w}\cdot\bol{w}}$ is the Euclidean norm of $\bol{w}$.
For every point $\bol{x}\in\Omega$ we want to find a neighborhood $V$ of $\bol{x}$ such that $\bol{\xi}$ is a symmetry of $\bol{w}$ in $V$ with $\di\bol{\xi}=0$.   
We decompose $\bol{\xi}$ on the local tangent basis $\lr{\p_{\ell},\p_{\psi},\p_{\theta}}$ as follows:
\begin{equation}
\bol{\xi}=\hat{h}\lr{\alpha\p_{\ell}+\beta\p_{\psi}+\gamma\p_{\theta}}~~~~{\rm in}~~U.\label{xi}
\end{equation}
Here, $\alpha$, $\beta$, and $\gamma$ are functions to be determined.
Substituting \eqref{w} and \eqref{xi} into equation \eqref{Sym3}, we obtain
\begin{equation}
\begin{split}
\mf{L}&_{\bol{\xi}}\bol{w}=\\
&\cu\left\{\hat{h}w^2\left[\lr{\beta\sin\theta-\alpha\cos\theta}\p_{\ell}\cp\p_{\psi}+\gamma\cos\theta\p_{\psi}\cp\p_{\theta}-\gamma\sin\theta\p_{\theta}\cp\p_{\ell}\right]\right\}.
\end{split}
\end{equation}
Next, observe that the Jacobian $J$ of the coordinate change $\lr{x,y,z}\mapsto\lr{\ell,\psi,\theta}$ is given by
\begin{equation}
J^{-1}=\nabla\ell\cdot\nabla\psi\cp\nabla\theta=\frac{1}{\p_{\ell}\cdot\p_{\psi}\cp\p_{\theta}}=\JI{w}=\hat{h}w^{2}.
\end{equation}
Hence, 
\begin{equation}
\mf{L}_{\bol{\xi}}\bol{w}=\cu\left[\lr{\beta\sin\theta-\alpha\cos\theta}\nabla\theta+\gamma\cos\theta\nabla\ell-\gamma\sin\theta\nabla\psi\right].\label{LS3}
\end{equation}
For the righ-hand side of \eqref{LS3} to vanish, it is necessary that the argument of the curl is the gradient of some function $\rho$.
Thus, we must solve the following system for the unknowns $\lr{\rho,\alpha,\beta,\gamma}$:
\begin{subequations}
\begin{align}
\frac{\p\rho}{\p\ell}&=\gamma\cos\theta,\label{r1}\\
\frac{\p\rho}{\p\psi}&=-\gamma\sin\theta,\label{r2}\\
\frac{\p\rho}{\p\theta}&=\beta\sin\theta-\alpha\cos\theta,\label{r3}\\
\di\bol{\xi}&=\hat{h}w^2\left[\frac{\p}{\p\ell}\lr{\frac{\alpha}{w^2}}+\frac{\p}{\p\psi}\lr{\frac{\beta}{w^2}}+\frac{\p}{\p\theta}\lr{\frac{\gamma}{w^2}}\right]=0.\label{r4}
\end{align}
\end{subequations}
The last equation follows from the identity 
\begin{equation}
\mf{L}_{\bol{\xi}}dx\w dy\w dz=\lr{\di\bol{\xi}}dx\w dy\w dz=\mf{L}_{\bol{\xi}}\frac{d\ell\w d\psi\w d\theta}{\hat{h}w^2}.
\end{equation}
On the other hand, in \cite{SatoB1} it is shown that
\begin{equation}
\hat{h}=\hat{h}\lr{\theta,L_{\theta}},~~~~L_{\theta}=\ell\cos\theta-\psi\sin\theta.
\end{equation}
This is a consequence of the fact that, locally, $\bol{w}$ can be expressed as $\bol{w}=\hat{h}^{-1}\nabla\theta\cp\nabla L_{\theta}$, and, since $\bol{w}$ is solenoidal, $\bol{w}\cdot\nabla\hat{h}=\hat{h}^{-1}\nabla\theta\cp\nabla L_{\theta}\cdot\nabla\hat{h}=0$.
It follows that
\begin{equation}
\frac{\p\hat{h}}{\p\ell}=\frac{\p\hat{h}}{\p L_{\theta}}\cos\theta,~~~~\frac{\p\hat{h}}{\p\psi}=-\frac{\p\hat{h}}{\p L_{\theta}}\sin\theta.
\end{equation}
Then, one sees that by setting $\rho=g{(L_{\theta},\theta)}$, 
where $g$ is an arbitrary function of the variables $L_{\theta}$ and $\theta$, both \eqref{r1} and \eqref{r2}
are satisfied with
\begin{equation}
 \gamma=\frac{\p g}{\p L_{\theta}}.
\end{equation}
From equation \eqref{r3}, we further have
\begin{equation}
\beta=\frac{\frac{\p g}{\p\theta}+\alpha\cos\theta}{\sin\theta}.
\end{equation}
This equaton will determine $\beta$ once $\alpha$ is known.
The function $\alpha$ can be obtained as a solution of \eqref{r4}, which, after some manipulations, reads as
\begin{equation}
\begin{split}
\frac{\p\alpha}{\p\ell}&+\frac{\cos\theta}{\sin\theta}\frac{\p\alpha}{\p\psi}-2\alpha\lr{\frac{\p\log w}{\p\ell}+\frac{\cos\theta}{\sin\theta}\frac{\p\log w}{\p\psi}}
+\\&-\frac{2}{\sin\theta}\frac{\p g}{\p\theta}\frac{\p\log w}{\p\psi}+\frac{1}{\sin\theta}\frac{\p^2 g}{\p\psi\p\theta}+\frac{\p}{\p\theta}\lr{\frac{\p g}{\p L_{\theta}}}-2\frac{\p g}{\p L_{\theta}}\frac{\p\log w}{\p\theta}=0.\label{r4_2}
\end{split}
\end{equation}
Notice that here $\p g/\p L_{\theta}$ is the partial derivative of $g$ with respect to $L_{\theta}$ when
$g$ is intended as a function of $\theta$ and $L_{\theta}$.
Equation \eqref{r4_2} is a first order linear partial differential equation for the variable $\alpha$.
Hence, solutions exist at least locally and can be obtained through the method of characteristics. 
We conclude that for all $\bol{x}\in\Omega$ there exist a neighborhood $V\subset\Omega$ of $\bol{x}$ and a vector field $\bol{\xi}\in TV$ such that
\begin{equation}
\bol{\xi}=\hat{h}\left[\alpha\p_{\ell}+\frac{\frac{\p g}{\p\theta}+\alpha\cos\theta}{\sin\theta}\p_{\psi}+\frac{\p g}{\p L_{\theta}}\p_{\theta}\right]~~~~{\rm in}~~V,\label{xi2}
\end{equation}
with $\alpha$ a solution of \eqref{r4_2} and 
\begin{equation}
\mf{L}_{\bol{\xi}}\bol{w}=0~~~~{\rm in}~~V.
\end{equation}
Observe that the symmetry $\bol{\xi}$ can always be chosen to be non-vanishing and not aligned with the solution $\bol{w}$ 
since, by construction, $\bol{w}\cp\bol{\xi}=\nabla g$ with $g$ an arbitrary function of $L_\theta$ and ${\theta}$.  
 
The second case occurs when $\nabla\chi\neq\bol{0}$ in some open subset $U\subset\Omega$.
Notice that $U$ must be open, since otherwise $\nabla\chi$ is discontinuous and the
smoothness of $\bol{w}\cp\lr{\cu\bol{w}}$ is violated. 
Furthermore, both $\bol{w}$ and $\cu\bol{w}$ are solenoidal.
Hence, recalling equation \eqref{Sym3}, we see that system \eqref{MFF} is equivalent to
\begin{equation}
\mf{L}_{\cu\bol{w}}\bol{w}=\bol{0},~~~~{\rm \di\bol{w}}=0~~~~{\rm in}~~\Omega.\label{MFFL} 
\end{equation}
Since $\nabla\chi\neq\bol{0}$ in $U$, it readily follows that $\bol{\xi}=\cu\bol{w}$ is the desired
symmetry in $U$. Notice that this remains true even if $\nabla\chi$ is allowed to vanish
over a set of measure zero in $U$ because $\cu\bol{w}$, and thus $\bol{\xi}$, are smooth in $U$. 
%the smoothness of $\bol{w}$ guarantees
%the continuity of $\BN{w}$ in $U$. 
Finally, an analogous argument holds for neighborhoods $U$ built around 
points situated at the boundary between regions with $\nabla\chi=\bol{0}$ and
regions with $\nabla\chi\neq\bol{0}$. Indeed, denoting by $U_0$ the open regions in $U$ with
$\nabla \chi=\bol{0}$ and setting $U_1=U-U_{0}$, we can choose $\bol{\xi}$ to be given
by \eqref{xi2} in $U_0$ and by $\cu\bol{w}$ in $U_1$.

%\item
%\end{enumerate} 
\end{proof}

\begin{remark}
Notice that \eqref{xi2} includes the trivial symmetry $\bol{\xi}=\bol{w}$, which can
be obtained for $g=0$ and $\alpha=\hat{h}^{-1}w^2\sin\theta$. 
\end{remark}

\begin{remark}
If $\nabla\chi\neq\bol{0}$ almost everywhere in $\Omega$, from equation
\eqref{MFFL}, it follows that $\cu\bol{w}$ is a symmetry of $\bol{w}$ in the whole $\Omega$.
\end{remark}

\noindent In the following examples we show how to
evaluate the directions of symmetry $\bol{\xi}$ 
for some specific Beltrami fields.

\begin{example} 
Set $\lr{\ell,\psi,\theta}=\lr{x,y,z}$, and consider the `minimal' ABC flow 
\begin{equation}
\bol{w}=\cos z\,\nabla y+\sin z\,\nabla x.\label{mABC}
\end{equation}
We refer the reader to \cite{Dombre} for a discussion on the properties of ABC flows. Notice that the vector field \eqref{mABC} satisfies $\di\bol{w}=0$ and $\hat{h}=w=1$. Set $g=g\lr{z}$. Recalling \eqref{r4_2}, $\alpha$ must satisfy the equation
\begin{equation}
\frac{\p\alpha}{\p x}+\frac{\cos z}{\sin z}\frac{\p\alpha}{\p y}
%+\frac{1}{\sin\theta}\frac{\p^2 g}{\p y\p z}+\frac{\p}{\p z}\lr{\frac{\p g}{\p L_{z}}}
=0.\label{a_c}
\end{equation}
The solution of \eqref{a_c} can be computed by the method of characteristics:
\begin{equation}
\alpha=p\lr{z,y-\frac{\cos z}{\sin z}x},
\end{equation}
with $p$ an arbitrary function of $z$ and $y-\frac{\cos z}{\sin z}x$. %+q\lr{z}$.%,
%and $q$ an arbitrary function of $z$.
Then, according to \eqref{xi2}, the direction of symmetry is given by the vector field
\begin{equation}
\bol{\xi}=p\,\p_x+\frac{\frac{\p g}{\p z}+p\cos z}{\sin z}\p_y.
\end{equation}
Since $p$ and $g$ are arbitrary, we can identify 
two important symmetries by setting $\lr{p,g}=\lr{1,-\sin z}$ and $\lr{p,g}=\lr{0,-\cos z}$ respectively.
These choices give the translational symmetries
\begin{equation}
\bol{\xi}_{1}=\p_x,~~~~\bol{\xi}_{2}=\p_y.
\end{equation}
\end{example}

\begin{example}
Let $\lr{r,z,\vartheta}=\lr{\sqrt{x^2+y^2},z,\arctan\lr{y/x}}$ denote a cylindrical coordinate system.
Set $\lr{\ell,\psi,\theta}=\lr{\vartheta,\log r,z}$, and consider the cylindrical Beltrami field
\begin{equation}
\bol{w}=\B{\vartheta}{\log r}{z}.
\end{equation}
Notice that $\di\bol{w}=0$, $\hat{h}=-1$, $w=1/r$. Set $g=g\lr{z}$.
Recalling \eqref{r4_2}, $\alpha$ must satisfy the equation
\begin{equation}
\frac{\p \alpha}{\p \vartheta}+\frac{\cos z}{\sin z}\frac{\p \alpha}{\p\log r}+2\frac{\alpha \cos z+\frac{\p g}{\p z}}{\sin z}=0.
\end{equation}
By the method of characteristics, one finds the solution
\begin{equation}
%\begin{split}
\alpha=\frac{1}{2}\frac{\sin z}{\cos z}p\lr{z,\log r-\frac{\cos z}{\sin z}\lr{\vartheta-q\lr{z}}}\exp\left\{2\frac{\cos z}{\sin z}\lr{q\lr{z}-\vartheta}\right\}-\frac{\frac{\p g}{\p z}}{\cos z}.
%\end{split}
\end{equation}
Here, $p$ is an arbitrary function of $z$ and $\log r-\frac{\cos z}{\sin z}\lr{\vartheta-q}$, and $q$ an arbitrary function of $z$.
The direction of symmetry is thus
\begin{equation}
\begin{split}
\bol{\xi}=&\lr{\frac{1}{2}\frac{\sin z}{\cos z}p\exp\left\{2\frac{\cos z}{\sin z}\lr{q-\vartheta}\right\}-\frac{\frac{\p g}{\p z}}{\cos z}}\p_{\vartheta}
\\&+\frac{1}{2}p\exp\left\{2\frac{\cos z}{\sin z}\lr{q-\vartheta}\right\}\p_{\log r}.
\end{split}
\end{equation}  
By setting $p=0$ and $g=-\sin z$, one obtains the
symmetry of rotation around the $z$-axis:
\begin{equation}
\bol{\xi}=\p_{\theta}.
\end{equation}
\end{example}

\begin{example}\label{ex3}
Set $\lr{\ell,\psi,\theta}=\lr{e^x\sin y,-e^x \cos y,z^2}$ and consider the Beltrami field
\begin{equation}
\bol{w}=-\cos\lr{ z^2}\,\nabla\lr{e^x\cos y}+\sin\lr{ z^2}\,\nabla\lr{e^x\sin y}.
\end{equation}
Notice that $\di\bol{w}=0$, $\hat{h}=2z$, $w=e^x=\sqrt{\ell^2+\psi^2}$. Set $g=g\lr{\theta}$.
Then, from $\eqref{r4_2}$, $\alpha$ must satisfy the equation
\begin{equation}
\frac{\p\alpha}{\p\ell}+\frac{\cos \theta}{\sin \theta}\frac{\p\alpha}{\p \psi}-\frac{2\alpha}{\ell^2+\psi^2}\lr{\ell+\frac{\cos \theta}{\sin \theta}\psi}-\frac{2}{\sin \theta}\frac{\p g}{\p \theta}\frac{\psi}{\ell^2+\psi^2}=0.
\end{equation}
By the method of characteristics, one can construct the solution
\begin{equation}
\alpha=\lr{\ell^2+\psi^2}p\lr{\theta,\psi-\frac{\cos\theta}{\sin\theta}\ell}+\frac{\frac{\p g}{\p\theta}}{\sin\theta}\frac{\ell\lr{\psi-\frac{\cos\theta}{\sin\theta}\ell}+\lr{\ell^2+\psi^2}\arctan\lr{\frac{\ell}{\psi}}}{\lr{\psi-\frac{\cos \theta}{\sin\theta}\ell}^2}.
\end{equation}
Here, $p$ is an arbitrary function of $\theta$ and $\psi-\frac{\cos\theta}{\sin\theta}\ell$.
The direction of symmetry is thus given by
\begin{equation}
\begin{split}
\bol{\xi}=&2\sqrt{\theta}\left[\lr{\ell^2+\psi^2}p+\frac{\frac{\p g}{\p\theta}}{\sin\theta}\frac{\ell\lr{\psi-\frac{\cos\theta}{\sin\theta}\ell}+\lr{\ell^2+\psi^2}\arctan\lr{\frac{\ell}{\psi}}}{\lr{\psi-\frac{\cos \theta}{\sin\theta}\ell}^2}\right]\p_{\ell}\\&+
2\sqrt{\theta}\frac{\frac{\p g}{\p\theta}+\left[\lr{\ell^2+\psi^2}p+\frac{\frac{\p g}{\p\theta}}{\sin\theta}\frac{\ell\lr{\psi-\frac{\cos\theta}{\sin\theta}\ell}+\lr{\ell^2+\psi^2}\arctan\lr{\frac{\ell}{\psi}}}{\lr{\psi-\frac{\cos \theta}{\sin\theta}\ell}^2}\right]\cos\theta}{\sin\theta}\p_{\psi}.
\end{split}
\end{equation}
\end{example}

In proposition \ref{prop1} we have shown that all solutions of \eqref{MFF} are locally symmetric.
Furthermore, the local symmetry is solenoidal, i.e. $\di\bol{\xi}=0$. 
However, it is important to stress that the symmetry $\bol{\xi}$ is not necessarily
an Euclidean isometry. We will discuss this fact in sections 3 and 4 with specific examples.
The purpose of the remaining part of this section is to obtain a local representation (Clebsch-like parametrization) valid for any solution of system \eqref{MFF} by exploiting the symmetry $\bol{\xi}$.
We have the following:

\begin{proposition}\label{prop2}
Let $\bol{w}\in C^{\infty}\lr{\Omega}$ denote a smooth solution of system \eqref{MFF} 
with smooth solenoidal symmetry $\bol{\xi}\in C^{\infty}\lr{\Omega}$, $\bol{\xi}\neq\bol{0}$ in $\Omega$. 
%Assume $\bol{w}\neq\bol{0}$ and $\cu\bol{w}\neq\bol{0}$ in $\Omega$.
Then, for every point $\bol{x}\in\Omega$, 
there exist a neighborhood $U$ of $\bol{x}$
and local curvilinear coordinates $\lr{x^1,x^2,x^3}\in C^{\infty}\lr{U}$ such that
\begin{equation}
\bol{w}=\nabla\Theta\cp\nabla x^3+w^3 \p_3~~~~{\rm in}~~U,\label{wsymloc}
\end{equation}
where $w^3=\bol{w}\cdot\nabla x^3=w^3\lr{x^1,x^2}$ and $\Theta=\Theta\lr{x^1,x^2}$ are
functions of $x^1$ and $x^2$.
Furthermore, if $\nabla\chi\neq\bol{0}$ and $\bol{\xi}=\cu\bol{w}$ in $U$, there exist a function $\Psi=\Psi\lr{x^1,x^2}$ of $x^1$ and $x^2$, and a function $\Phi=\Phi\lr{x^1,x^2,x^3}$ such that
\begin{subequations}
\begin{align}
%\nabla\Theta\cp\nabla\Psi&=\nabla x^1\cp\nabla x^2,\\
\Psi_1\Theta_2-\Psi_2\Theta_1&=1,\label{PsiTheta}\\
-\frac{\nabla\Theta\cp\lr{\nabla\Phi\cp\nabla\Theta}}{\abs{\nabla\Theta}^2}\cdot\nabla\lr{\frac{\nabla\Theta\cdot\nabla\Phi}{\abs{\nabla\Theta}^2}}&=1.\label{GGSE}
\end{align}
\end{subequations}
and
\begin{equation}
\bol{w}=\Psi\,\nabla\Theta+\nabla\Phi,~~~~\chi=-\Theta~~~~{\rm in}~~U.
\end{equation}
\end{proposition}
  
\begin{proof}
Since the direction of symmetry $\bol{\xi}$ is smooth, solenoidal, and non-vanishing in $\Omega$, 
from the Lie-Darboux theorem for every point $\bol{x}\in\Omega$ there exist 
a neighborhood $U$ of $\bol{x}$ and smooth functions $\sigma$ and $\tau$ such that 
\begin{equation}
\bol{\xi}=\nabla \sigma\cp\nabla \tau~~~~{\rm in}~~U.
\end{equation}
We set $x^1=\sigma$ and $x^2=\tau$. 
The intersections of the level sets of $x^1$ and $x^2$ define curves whose tangent
vector $\nabla\sigma\cp\nabla\tau$ can be used as the tangent vector $\p_3$ of the third coordinate $x^3$, i.e.
\begin{equation}
\p_3=\nabla\sigma\cp\nabla\tau~~~~{\rm in}~~U.
\end{equation}
In this way, we have specified the local coordinate system $\lr{x^1,x^2,x^3}\in C^{\infty}\lr{U}$. 
Notice that the Jacobian of the transformation from the standard Cartesian frame to the new coordinates $\lr{x,y,z}\mapsto\lr{x^1,x^2,x^3}$ is given by 
\begin{equation}
J=\p_1\cdot\p_2\cp\p_3=\frac{1}{\nabla x^3\cdot\nabla x^1\cp\nabla x^2}=\frac{1}{\nabla x^3\cdot \p_3}=1.
\end{equation}
Next, observe that in $U$ the condition of symmetry reads as
\begin{equation}
\mf{L}_{\p_3}\bol{w}=\lr{\p_3\cdot\nabla}w^i \p_i-\lr{\bol{w}\cdot\nabla}\p_3=\frac{\p w^i}{\p x^3}\p_i=\bol{0}.
\end{equation} 
This implies
\begin{equation}
w^i=w^i\lr{x^1,x^2},~~~~i=1,2,3.\label{wi}
\end{equation}
Furthermore, the vector field $\bol{w}$ is solenoidal:
\begin{equation}
\begin{split}
\mf{L}_{\bol{w}}dx \w dy \w dz&=\lr{\di\bol{w}}\,dx\w dy \w dz\\&=\lr{\frac{\p w^1}{\p x^1}+\frac{\p w^2}{\p x^2}+\frac{\p w^3}{\p x^3}}\,dx^1\w dx^2\w dx^3=0.
\end{split}
\end{equation}
Using \eqref{wi}, this gives the condition
\begin{equation}
\frac{\p w^1}{\p x^1}=-\frac{\p w^2}{\p x^2}.
\end{equation}
Hence, if we define the stream function
\begin{equation}
\Theta\lr{x^1,x^2}=-\int_{a}^{x^1}w^2\lr{s^1,x^2}\, ds^1+\int_{b}^{x^2} w^1\lr{a,s^2}\, ds^2,
\end{equation}
with $\lr{a,b,c}$ a point in $U$, it follows that
\begin{equation}
w^1=\frac{\p\Theta}{\p x^2},~~~~w^2=-\frac{\p \Theta}{\p x^1},
\end{equation}
and also
\begin{equation}
\bol{w}=w^i \p_i=\Theta_2 \p_1-\Theta_1 \p_2+w^3\p_3=\nabla\Theta\cp\nabla x^3+w^3\p_3.
\end{equation}
We have thus obtained equation \eqref{wsymloc}.

Now suppose that $\nabla\chi\neq\bol{0}$ in $U$ and choose $\bol{\xi}=\cu\bol{w}$ to be 
the symmetry given by proposition \ref{prop1}. Then, 
\begin{equation}
\cu\bol{w}=\nabla x^1\cp\nabla x^2~~~~{\rm in}~~U.\label{cuw1}
\end{equation}
The first equation of system \eqref{MFF} reads as
\begin{equation}
\begin{split}
\BN{w}&=\lr{\nabla\Theta\cp\nabla x^3+w^3\p_3}\cp\lr{\nabla x^1\cp\nabla x^2}=\\
&=\left[\lr{\nabla\Theta\cp\nabla x^3}\cdot\nabla x^2\right]\nabla x^1-
\left[\nabla\Theta\cp\nabla x^3\cdot\nabla x^1\right]\nabla x^2\\
&=-\nabla\Theta.\label{BNw1}
\end{split}
\end{equation}
Hence, $\chi=-\Theta$. 
Next, we look for a function $\Psi\lr{x^1,x^2}$ satisfying
\begin{equation}
\Psi_1\Theta_2-\Psi_2\Theta_1=1\label{cuw1_2}
\end{equation} 
This equation is a first order linear partial differential equation
for the variable $\Psi$ and can be solved in a neighborhood $V\subseteq U$ by the method of characteristics.
Without loss of generality, we can restrict our domain $U$ so that $U=V$.  
From equations \eqref{cuw1} and \eqref{cuw1_2} we thus obtain
%\begin{equation}
%\nabla\cp\bol{w}=\nabla\Psi\cp\nabla\Theta~~~~{\rm in}~~U.\label{cuw2}
%\end{equation}
%Furthermore, from $\nabla\Theta\cdot\cu\bol{w}=0$,  
%it follows that there exists a function $\Psi\lr{x^1,x^2}$ such that
\begin{equation}
\nabla\cp\bol{w}=\nabla\Psi\cp\nabla\Theta~~~~{\rm in}~~U.\label{cuw2}
\end{equation}
%Combining equations \eqref{cuw1} and \eqref{cuw2}, we have
%\begin{equation}
%\lr{\Psi_1\Theta_2-\Psi_2\Theta_1-1}\nabla x^1\cp\nabla x^2=\bol{0},
%\end{equation}
%which is equation \eqref{PsiTheta}.
Next, decurling equation \eqref{cuw2} gives
\begin{equation}
\bol{w}=\Psi\,\nabla\Theta+\nabla\Phi,\label{wClebsch3}
\end{equation}
for some function $\Phi$. 
On the other hand, $\bol{w}$ is orthogonal to $\nabla\Theta$. 
Therefore, the projection of $\bol{w}$ on the space orthogonal to $\nabla\Theta$ leaves $\bol{w}$ unchanged, i.e.
\begin{equation}
\bol{w}=\frac{\nabla\Theta\cp\lr{\bol{w}\cp\nabla\Theta}}{\abs{\nabla\Theta}^2}=\frac{\nabla\Theta\cp\lr{\nabla\Phi\cp\nabla\Theta}}{\abs{\nabla\Theta}^2}.
\end{equation}
In the last passage equation \eqref{wClebsch3} was used.
Finally, substituting this expression for $\bol{w}$ in the first equation of system \eqref{MFF} gives:
\begin{equation}
\begin{split}
\BN{w}&=\bol{w}\cp\left[\cu\lr{\nabla\Phi-\frac{\nabla\Theta\cdot\nabla\Phi}{\abs{\nabla\Theta}^2}\nabla\Theta}\right]\\
&=-\frac{\nabla\Theta\cp\lr{\nabla\Phi\cp\nabla\Theta}}{\abs{\nabla\Theta}^2}\cp\left[\nabla\lr{\frac{\nabla\Theta\cdot\nabla\Phi}{\abs{\nabla\Theta}^2}}\cp\nabla\Theta\right]\\
&=\left[\frac{\nabla\Theta\cp\lr{\nabla\Phi\cp\nabla\Theta}}{\abs{\nabla\Theta}^2}\cdot
\nabla\lr{\frac{\nabla\Theta\cdot\nabla\Phi}{\abs{\nabla\Theta}^2}}\right]\nabla\Theta.
\end{split}
\end{equation}
Recalling equation \eqref{BNw1}, we obtain equation \eqref{GGSE}.
\end{proof}	

Euquation \eqref{GGSE} corresponds to the
generalized Grad-Shafranov equation derived by Weitzner (equation (28c) of \cite{Weitzner}) expressed in terms of the local coordinates derived from the symmetry. In practice, equation \eqref{GGSE} is applied as follows.
%Suppose that we have a domain $\Omega$ with symmetry $\bol{\xi}$.
%We say that $\Omega$ is symmetric with respect to $\bol{\xi}$ when
%$\bol{\xi}\cdot\bol{n}=0$ on the boundary $\p\Omega$.
In order to construct a solution $\bol{w}$ of system \eqref{MFF} with
a given solenoidal symmetry $\bol{\xi}$, first we determine the associated local coordinate system $\lr{x^1,x^2,x^3}$. Then, we assign $\Theta$ as a given function of $\lr{x^1,x^2}$ and determine the corresponding function $\Psi$ from equation \eqref{PsiTheta}. Finally, we look for solutions $\Phi$ of 
equation \eqref{GGSE}. Notice that equation \eqref{GGSE}, which is a second order nonlinear partial differential equation, is not elliptic (see \cite{Gilbarg} for the definition of second order elliptic partial differential operator) because
the coefficient matrix $a^{ij}=\left[\frac{\nabla\Theta\cp\lr{\nabla\Phi\cp\nabla\Theta}}{\abs{\nabla\Theta}^2}\right]^i\frac{\Theta_j}{\abs{\nabla\Theta}^2}$ is not positive definite (if $\bol{v}\in\mathbb{R}^3$ is a vector perpendicular to $\nabla\Theta\lr{\bol{x}}$ at a point $\bol{x}$, $a^{ij} v_j=0$ there). Hence, the existence of a solution is not guaranteed. If a solution is not found, one can repeat the procedure
by changing the ansatz on the stream function $\Theta$. 
Notice that $\bol{\xi}$ could be chosen to be a (not necessarily Euclidean) solenoidal symmetry of a topological torus $T$, i.e. $\bol{\xi}$ could be a solenoidal vector field tangent to nested flux (toroidal) surfaces $\Theta=$constant. A solution constructed according to the procedure described above
will not exhibit, in general, Euclidean symmetries. 
Furthermore, such solutions will satisfy the boundary condition \eqref{BC}
because $\bol{n}=\nabla\Theta/\abs{\nabla\Theta}$ on $\p T$.

%\section{Euclidean symmetries}

\section{Beltrami fields without continuous Euclidean symmetries}

The purpose of the present section is to show that
system \eqref{MFF} admits nontrivial Beltrami field solutions
without continuous Euclidean symmetries. %At this end, we first
%need to review some facts about Euclidean symmetries. 
Euclidean symmetries are the set of transformations of Euclidean space
that preserve the Euclidean distance between points. 
These transformations comprise translations, rotations, reflections,
and their combination. They define the Euclidean group $E\lr{3}$. Translations 
%(the group operation of the translational group $T\lr{3}$) 
and rotations 
%(the group operation of the special orthogonal group $SO\lr{3}$) 
are continuous transformations, 
called direct isometries, that do not change the orientation of space; together they form the special Euclidean group $SE\lr{3}$. 
Reflections, on the other hand, are indirect isometries that flip
the orientation of space. Furthermore, reflections are discrete transformations 
because the orbit of a point in the metric space under the isometry forms a discrete set.  
As mentioned in the introduction, if one postulates 
the existence of an ignorable coordinate $x^3$ such that the components $w^i$, the pressure field $\chi$,
and the metric tensor $g_{ij}$ are all independent of $x^3$,
it is known that system \eqref{MFF} can be reduced to a single 
nonlinear elliptic second order partial differential equation, the Grad-Shafranov equation \eqref{GSE}. For completeness,
the derivation of the Grad-Shafranov equation is given in appendix A (we also refer the reader to \cite{Eden1,Eden2} on this point).
The requirement of an ignorable coordinate restricts the class of symmetries to direct isometries,
so that the Grad-Shafranov equation does not apply to reflectional symmetries. 
This fact can be seen explicitly by noting that the condition $\p_3g_{ij}=0$ is equivalent
to the symmetry condition $\mf{L}_{\p_3}g_{ij}dx^i\otimes dx^j=0$ for the metric tensor, 
and that $\p_3=\bol{\xi}$ is the direction of symmetry. Then, the following holds:

\begin{proposition}\label{prop3}
The only symmetry preserving the Euclidean metric tensor $g_{ij}=\delta_{ij}$ of $\mathbb{R}^3$ is spanned by the vector field
\begin{equation}
\bol{\xi}_{E}=\bol{a}+\bol{b}\cp\bol{x},
\end{equation} 
with $\bol{a},\bol{b}\in\mathbb{R}^3$ arbitrary constant vectors.
\end{proposition}

\begin{proof}
First, consider a vector field $\bol{\xi}\in T\Omega$ and evaluate the Lie-derivative of the metric tensor
\begin{equation}
\begin{split}
\mf{L}_{\bol{\xi}}g_{ij}dx^i \otimes dx^j&=\sum_{i}\left[\mf{L}_{\bol{\xi}}dx^{i}\otimes dx^i+dx^i \otimes \mf{L}_{\bol{\xi}}d x^i\right]\\
&=\sum_{i}\left[\lr{di_{\bol{\xi}}dx^i}\otimes dx^i+dx^i\otimes \lr{di_{\bol{\xi}}dx^i}\right]\\
&=\sum_{i}\left[d\xi^i \otimes dx^i +dx^i\otimes d\xi^i\right]\\
&=\sum_{ij}\lr{\frac{\p\xi^i}{\p x^j}+\frac{\p \xi^j}{\p x^i}}dx^i\otimes dx^j.
\end{split}
\end{equation}
Suppose that the Lie derivative vanishes. 
Setting $i=j$ gives $\frac{\p \xi^i}{\p x^i}=0$, $i=1,2,3$. Hence, $\xi^1=\xi^1\lr{y,z}$, $\xi^2=\xi^2\lr{z,x}$, and $\xi^3=\xi^3\lr{x,y}$.
For $i\neq j$, one obtains the conditions
\begin{equation}
\frac{\p\xi^1}{\p y}=-\frac{\p\xi^2}{\p x},~~~~\frac{\p\xi^1}{\p z}=-\frac{\p\xi^3}{\p x},~~~~\frac{\p\xi^2}{\p z}=-\frac{\p\xi^3}{\p y}.\label{xicond}
\end{equation}
Next, observe that
\begin{equation}
\frac{\p^2\xi^1}{\p y^2}=-\frac{\p^2\xi^2}{\p x\p y}=0,~~~~\frac{\p^2\xi^1}{\p z^2}=-\frac{\p^2\xi^3}{\p z\p x}=0.
\end{equation}
Integrating each equation gives
\begin{equation}
\xi^1=u_{1}\lr{z}y+v_{1}\lr{z},~~~~\xi^1=f_{1}\lr{y}z+g_{1}\lr{y}.\label{xi1__}
\end{equation}
Here, $u_{1}$ and $v_{1}$ are functions of $z$, and $f_{1}$ and $g_{1}$ functions of $y$ to be determined.
Deriving these equations with respect to $y$ leads to
\begin{equation}
\frac{\p \xi^1}{\p y}=u_{1}=\frac{\p f_{1}}{\p y}z+\frac{\p g_{1}}{\p y}.
\end{equation}
It follows that
\begin{equation}
f_{1}=k_{1}y+b_{1}^2,~~~~g_{1}=b_{1}^3y+a^{1},~~~~u_{1}=k_{1}z+b_{1}^3,~~~~k_{1},b_{1}^2,b_{1}^3,a^{1}\in\mathbb{R}.
\end{equation}
Substituting these expressions in \eqref{xi1__} we obtain
\begin{equation}
\xi^1=k_{1}yz+b_1^2　z+b_{1}^3 y+a^1.
\end{equation}
In a similar manner, one can show that
\begin{subequations}
\begin{align}
\xi^2=k_{2}zx+b_2^1 z+b_2^3 x +a^2,~~~~k_{2},b_2^1,b_2^3,a^2\in\mathbb{R},\\
\xi^3=k_{3}xy+b_3^1 y+b_3^2 x +a^3,~~~~k_{3},b_3^1,b_3^2,a^3\in\mathbb{R}.
\end{align}
\end{subequations}
Plugging these expressions into the conditions \eqref{xicond}, we conclude that
\begin{equation}
k_1=k_2=k_3=0,~~~~b_{1}^3=-b_{2}^3=-b^3,~~~~b_1^2=-b_3^2=b^2,~~~~b_2^1=-b_3^1=-b^1.
\end{equation}
\end{proof}

\begin{remark}
Notice that $\bol{a}$ represents translations, while $\bol{b}\cp\bol{x}$ rotations.
Indeed, the cross product $\bol{b}\cp\bol{x}$ can always be uniquely represented by the action of an antisymmetric matrix $B=-B^T$ as $B\bol{x}$. Then, 
\begin{equation}
\begin{bmatrix}\bol{\xi}_{E}\\0\end{bmatrix}=A v,~~~~A=\begin{bmatrix} B&\bol{a}\\\bol{0}^T&0\end{bmatrix},~~~~v=\begin{bmatrix}\bol{x}\\1\end{bmatrix}.
\end{equation}
The $4\cp 4$ matrix $A$ spans the Lie algebra $\mf{se}\lr{3}$ of the special Euclidean group $SE\lr{3}$, i.e. the tangent space of $SE\lr{3}$ at the identity.
The vector field $\bol{\xi}_{E}=\bol{a}+\bol{b}\cp\bol{x}$ is called a Killing vector field of the Euclidean metric. 
\end{remark}

From proposition \ref{prop3}, it is now clear that a solution $\bol{w}$ of system \eqref{MFF} will not possess
continuous Euclidean symmetries provided that
\begin{equation}
\mf{L}_{\bol{\xi_E}}\bol{w}\neq\bol{0}~~~~\forall \bol{a},\bol{b}\in\mathbb{R}^3:\bol{\xi}_E=\bol{a}+\bol{b}\cp\bol{x}\neq\bol{0}.\label{SE3asym}
\end{equation}
If one can find a solution $\bol{w}$ satisfying \eqref{SE3asym}, 
then such solution cannot be obtained as a solution of the Grad-Shafranov equation \eqref{GSE}. 
%However, equation \eqref{SE3asym} does not preclude the possibility that
%$\bol{w}$ is symmetric with respect to an Euclidean isometry involving reflections.
%Therefore, to exclude all Euclidean symmetries, a more general condition is necessary.
%Given a point $\bol{x}\in\mathbb{R}^3$, the most general Euclidean isometry
%can be represented as the transformation mapping $\bol{x}$ to the point $\bol{x}'=Q\bol{x}+\bol{a}$,
%where $Q$ is a real valued orthogonal matrix such that $QQ^T=I$, with $I$ the identity matrix, and $\bol{a}\in\mathbb{R}^3$ a constant vector. Pure rotations correspond to the case $\det\lr{Q}=1$, 
%pure translation to the case $Q=I$, while transformation involving reflections
%correspond to the case $\det\lr{Q}=-1$. 
%Then, a solution $\bol{w}$ will not exhibit any Euclidean symmetry whenever
%\begin{equation}
%\bol{w}\lr{Q\bol{x}+\bol{a}}\neq\bol{w}\lr{\bol{x}}~~~~\forall\bol{a}\in\mathbb{R}^3, Q\in O\lr{3}: Q\bol{x}+\bol{a}\neq\bol{x}.\label{E3asym}
%\end{equation} 
%Here, $O\lr{3}$ denotes the orthogonal group in three dimensions. 

We are ready to construct nontrivial Beltrami field solutions of system \eqref{MFF}
without continuous Euclidean symmetries. To achieve this goal, we apply the Clebsch-like parametrization 
of Beltrami fields derived in \cite{SatoB1}. 
In particular, our aim is to find a coordinate system $\lr{x^1,x^2,x^3}$
satisfying system \eqref{MT}. Then, we know that a vector field in the form \eqref{LBF}
is a nontrivial Beltrami field. %with proportionality coefficient $\hat{h}=$ (see \cite{SatoB1}).
The desired coordinate system can be characterized as follows:

\begin{mydef}
Let $\lr{x^1,x^2,x^3}\in C^{\infty}\lr{\Omega}$ denote a curvilinear coordinate system
in a domain $\Omega\subset\mathbb{R}^3$. The coordinate system is admissible in $\Omega$ if it satisfies
system \eqref{MT} in $\Omega$. 
We denote by $\mc{A}_{\Omega}$ the set of admissible curvilinear coordinate systems in $\Omega$,
and by $\mc{A}^{\perp}_{\Omega}$ the set of orthogonal admissible coordinate systems in $\Omega$.
Evidently, $\mc{A}^{\perp}_{\Omega}\subset\mc{A}_{\Omega}$.
\end{mydef}

\noindent Similarly, the notion of admissible proportionality coefficient is useful:

\begin{mydef}
Let ${f}\in C^{\infty}\lr{\Omega}$ be a smooth function in a domain $\Omega\subset\mathbb{R}^3$. 
${f}$ is an admissible proportionality coefficient in $\Omega$ if there exists a smooth nontrivial Beltrami field solution $\bol{w}\in C^{\infty}\lr{\Omega}$ of equation \eqref{MFF} with proportionality coefficient $\hat{h}=f$. The set of admissible
proportionality coefficients $\hat{h}\in C^{\infty}\lr{\Omega}$, $\hat{h}\neq 0$ in $\Omega$ is denoted by $\mf{A}_{\Omega}$.
\end{mydef}

\noindent The following propositions narrows down the class of admissible proportionality coefficients $\mf{A}_{\Omega}$
that will be allowed in the construction of the sought Beltrami field solutions once an admissible coordinate system $\lr{x^1,x^2,x^3}\in\mc{A}_{\Omega}$ has been found.

\begin{proposition}\label{prop4}
Let $\Omega\subset\mathbb{R}^3$ denote a bounded domain. 
%Let $\lr{x^1,x^2,x^3}\in C^{\infty}\lr{\Omega}$ be a curvilinear coordinate system in a domain $\Omega\subset\mathbb{R}^3$. Then, 
The set $\mf{A}_\Omega$ of admissible proportionality coefficients $\hat{h}\in C^{\infty}\lr{\Omega}$, $\hat{h}\neq 0$, such that equation \eqref{MFF} has a smooth nontrivial Beltrami field solution in $\Omega$ satisfies
\begin{subequations}
\begin{align}
\mf{A}&\supset\left\{\hat{h}\in C^{\infty}\lr{\Omega}~\rvert~ \hat{h}=f\lr{x^3}\sqrt{g^{33}},~\lr{x^1,x^2,x^3}\in \mc{A}_{\Omega}^{\perp}\right\},\label{A_1}\\
\mf{A}&\subset\left\{\hat{h}\in C^{\infty}\lr{\Omega}~\rvert~ \hat{h}=\hat{h}\lr{x^3,L_{x^3}},~\lr{x^1,x^2,x^3}\in \mc{A}_{U},~\forall\bol{x}\in\Omega~ \exists U\subset\Omega\right\}.\label{A_2}
\end{align}
\end{subequations}
Here, $f\in C^{\infty}\lr{\ov{\Omega}}$, is an arbitrary smooth function of the coordinate $x^3$ such that $f\neq 0$, and $L_{x^3}=x^1\cos x^3-x^2\sin x^3$.
\end{proposition}

\begin{proof}
First, we prove \eqref{A_1}. Define
\begin{equation}
\sigma\lr{x^3}=\int f\lr{x^3}\,dx^3.
\end{equation}
This integral is well defined since the function $f$ is, by hypothesis, smooth in the compact domain $\ov{\Omega}$. 
Then, given an admissible orthogonal coordinate system $\lr{x^1,x^2,x^3}\in \mc{A}_{\Omega}^{\perp}$, it is not difficult to verify that the vector field
\begin{equation}
\bol{w}=\cos\left[{\sigma\lr{x^3}}\right]\,\nabla x^2+\sin\left[{\sigma\lr{x^3}}\right]\,\nabla x^1,\label{wA1}
\end{equation}
is a Beltrami field in $\Omega$ with proportionality coefficient
\begin{equation}
\hat{h}=\frac{f\lr{x^3}}{g^{11}\sqrt{\abs{g}}}=f\lr{x^3}\sqrt{g^{33}}.\label{h_}
\end{equation}
Here, $\abs{g}=\lr{\p_1\cdot\p_2\cp \p_3}^2$ denotes the determinant of the covariant metric tensor $g_{ij}$. 
Indeed, taking the curl of \eqref{wA1} gives
\begin{equation}
\cu\bol{w}=f\lr{\sin \sigma \,\nabla x^2\cp\nabla x^3+\cos\sigma\,\nabla x^3\cp\nabla x^1}.
\end{equation}
Noting that $\nabla x^1\cdot\nabla x^2\cp\nabla x^3=1/\sqrt{\abs{g}}$, and $g^{11}=g^{22}$ from system \eqref{MT} and orthogonality of the coordinates, we further have
\begin{equation}
\begin{split}
\cu\bol{w}&=\frac{f}{\sqrt{\abs{g}}}\lr{\sin \sigma \,\p_{1}+\cos\sigma\,\p_{2}}\\&=\frac{f}{g^{11}\sqrt{\abs{g}}}\lr{\sin\sigma\,\nabla x^1+\cos\sigma\,\nabla x^2}=\frac{f}{g^{11}\sqrt{\abs{g}}}\,\bol{w}.
\end{split}
\end{equation}
Observe that $1/\sqrt{\abs{g}}=g^{11}\sqrt{g^{33}}$, which gives equation \eqref{h_}.
Next, we prove \eqref{A_2}. We must show that, given a solenoidal Beltrami field with proportionality coefficient $\hat{h}\neq 0$ in $\Omega$, there exists
an admissible coordinate system $\lr{x^1,x^2,x^3}\in \mf{C}_{\Omega}$ such that $\hat{h}=\hat{h}\lr{x^3,L_{x^3}}$.
This directly follows from theorem 1 of \cite{SatoB1}. Indeed, if $\bol{w}$ is a smooth nontrivial Beltrami field in $\Omega$, 
$\forall\bol{x}\in\Omega$ there exist a neighborhood $U\subset\Omega$ of $\bol{x}$ and an admissible coordinate system $\lr{x^1,x^2,x^3}\in \mc{A}_{U}$
such that
\begin{equation}
\bol{w}=\cos x^3\,\nabla x^2+\sin x^3\, \nabla x^1~~~~{\rm in}~~U.
\end{equation}
Since $\bol{w}$ is a Beltrami field, this implies
\begin{equation}
\bol{w}=\hat{h}^{-1}\nabla\cp\bol{w}=\hat{h}^{-1}\nabla x^3\cp\nabla L_{x^3}.
\end{equation}
Recalling that $\di\bol{w}=0$, we have $\nabla\hat{h}\cdot\nabla x^3\cp\nabla L_{x^3}=0$. Hence, $\hat{h}=\hat{h}\lr{x^3,L_{x^3}}$.
\end{proof}

We expect that, if the coordinate system $\lr{x^1,x^2,x^3}$ is built in such a way that it does not possess Euclidean symmetries,
the same should hold for the corresponding Beltrami field \eqref{LBF}.
System \eqref{MT} greatly simplifies if the coordinate system $\lr{x^1,x^2,x^3}$ is assumed orthogonal.
In such case, $g^{ij}=0$ for all $i\neq j$. Then, \eqref{MT} reduces to
\begin{subequations}\label{MTR}
\begin{align}
g^{11}&=g^{22},\\
\Delta x^2\cos x^3&=\Delta x^1 \sin x^3.
\end{align}
\end{subequations} 
These conditions suggest that an intimate relationship exists between Beltrami fields
and two-dimensional harmonic conjugate functions (this fact is discussed in \cite{SatoB2}). 
Indeed, system \eqref{MTR} can be satisfied by setting $x^3=\sigma\lr{z}$, 
$x^1=u\lr{x,y}$, and $x^2=v\lr{x,y}$, with $\sigma$ a function of the
Cartesian coordinate $z$, and $u$ and $v$ harmonic conjugate functions in the $x$-$y$ plane such that
\begin{equation}
u_x=v_y,~~~~u_y=-v_x.
\end{equation}
Thus, the idea is to look for a Beltrami field solution without Euclidean symmetries in the form
\begin{equation}
\bol{w}=\cos\left[\sigma\lr{z}\right]\,\nabla v\lr{x,y}+\sin\left[\sigma\lr{z}\right]\,\nabla u\lr{x,y}.
\end{equation}
On this regard, we have the following:

\begin{proposition}\label{prop5}
There exist nontrivial Beltrami field solutions of system \eqref{MFF} without continuous Euclidean symmetries, 
i.e. smooth solenoidal vector fields $\bol{w}\in C^{\infty}\lr{\Omega}$ with $\JI{w}\neq 0$ solving \eqref{BF} 
in some bounded domain $\Omega\subset\mathbb{R}^3$ and satisfying the condition \eqref{SE3asym} in $\Omega$.
\end{proposition}

\begin{proof}
It is sufficient to provide an example of nontrivial Beltrami field without continuous Euclidean symmetries.
We claim that the Beltrami field
\begin{equation}
\bol{w}=-\cos\lr{e^z}\,\nabla\lr{e^x\cos y}+\sin\lr{e^z}\nabla\lr{e^x\sin y},\label{Basym}
\end{equation}
corresponding to the choice $\lr{x^1,x^2,x^3}=\lr{u,v,\sigma}=\lr{e^x\sin y, -e^x \cos y,e^z}$, does not possess continuous Euclidean symmetries. Furthermore, it is smooth and nontrivial in any bounded domain $\Omega\subset\mathbb{R}^3$. 

Let us verify that \eqref{Basym} is not symmetric under a combination of translations and rotations by
evaluating \eqref{SE3asym}. After some manipulations, equation \eqref{Basym} can be rearranged as
\begin{equation}
\bol{w}=e^x \sin\lr{y+e^z}\,\nabla y-e^x \cos\lr{y+e^z}\,\nabla x.
\end{equation}
Next, observe that
\begin{equation}
\begin{split}
\mf{L}_{\lr{\bol{a}+\bol{b}\cp\bol{x}}}\bol{w}&=\left[\lr{\bol{a}+\bol{b}\cp\bol{x}}\cdot\nabla\right]\bol{w}-\lr{\bol{w}\cdot\nabla}\lr{\bol{b}\cp\bol{x}}
\\&=\left[\lr{\bol{a}+\bol{b}\cp\bol{x}}\cdot\nabla\right]\bol{w}-\bol{b}\cp\bol{w}.\label{LE}
\end{split}
\end{equation}
Let $\lr{w_x,w_y,w_z}$, $\lr{a_x,a_y,a_z}$, and $\lr{b_x,b_y,b_z}$ denote the Cartesian components of the vector fields $\bol{w}$, $\bol{a}$, and $\bol{b}$. Projecting equation \eqref{LE} along $\nabla x$ and substituting \eqref{Basym}, we obtain
\begin{equation}
\begin{split}
\nabla x\cdot\mf{L}_{\lr{\bol{a}+\bol{b}\cp\bol{x}}}\bol{w}=&-\left[\lr{\bol{a}+\bol{b}\cp\bol{x}}\cdot\nabla\right]\left[e^x\cos\lr{y+e^z}\right]+b_z e^x\sin\lr{y+e^z}\\
=&e^x\sin\lr{y+e^z}\left[a_y+b_z \lr{1+x}-b_x z+e^z\lr{a_z+b_x y-b_y x}\right]\\
&-e^x\cos\lr{y+e^z}\lr{a_x+b_y z-b_z y}.\label{LBasym}
\end{split}
\end{equation}
The vector field \eqref{Basym} possesses a continuous Euclidean symmetry if there exist a vector field $\bol{\xi}_{E}=\bol{a}+\bol{b}\cp\bol{x}\neq\bol{0}$ such that the right-hand side of the equation above vanishes in some domain $\Omega\subset\mathbb{R}^3$. 
Choose $\Omega$ so that it contains a finite area of the surface $S=\left\{\bol{x}\in\mathbb{R}^3:y=-e^{z}\right\}$. On the surface, the condition that
the right-hand side of \eqref{LBasym} vanishes reduces to
\begin{equation}
a_x+b_y z+b_z e^z=0~~~~{\rm on}~~S.
\end{equation}
This immediately leads to $a_x=b_y=b_z=0$. 
Then, the right-hand side of \eqref{LBasym} vanishes if
\begin{equation}
a_y+b_x \lr{e^z y-z}+a_z e^z=0~~~~{\rm in}~~\Omega.
\end{equation}
%Next, adjust $\Omega$ so that it also contains a finite area of the surface $S_2=\left\{\bol{x}\in\mathbb{R}^3:y=-e^{z}+\pi/2\right\}$.
%On the surface, the condition that
%the right-hand side of \eqref{LBasym} vanishes reduces to
%\begin{equation}
%a_y+a_z e^z+b_x \lr{\frac{\pi}{2}e^z-e^{2z}-z}=0~~~~{\rm on}~~S_2.
%\end{equation}
This condition can be satisfied only when $a_y=a_z=b_x=0$. Thus, $\bol{a}+\bol{b}\cp\bol{x}=\bol{0}$ and no
continuous Euclidean symmetry exists.
Finally, it is clear that the vector field \eqref{Basym} is smooth in any bounded domain $\Omega\subset\mathbb{R}^3$,  
and nontrivial in $\Omega$ since $\hat{h}=e^z$ and $w^2=e^{2x}$ imply $\JI{w}=\hat{h} w^2=e^{2x+z}\neq 0$ in $\Omega$.
%To exclude isometries that involve reflections, we now must verify equation \eqref{E3asym}.  
%Observe that we could have checked \eqref{E3asym} directly. However, 
%one is usually interested in direct isometries that enable the reduction of system \eqref{MFF}
%to the Grad-Shafranov equation. In this case, the relevant condition is given by \eqref{SE3asym}.
%Let $x^j$, $j=1,2,3,$, denote the Cartesian coordinates $\lr{x,y,z}$ and $Q^{i}_{\,j}$, $i,j=1,2,3$, 
%the Cartesian components of the candidate orthogonal matrix $Q$. 
%Consider the projection of equation \eqref{E3asym} along $\nabla x$. We have:
%\begin{equation}
%w_x\lr{Q\bol{x}+\bol{a}}=w_{x}\lr{\bol{x}}.
%\end{equation}
%Recalling \eqref{Basym}, this is equivalent to demanding that
%\begin{equation}
%\exp\left\{Q^{1}_{\,j}x^j+a_x\right\}\cos\lr{Q^{2}_{\,j}x^j+a_y+\exp\left\{Q^{3}_{\,j}x^j+a_z\right\}}=e^x\cos\lr{y+e^z}.
%\end{equation}

Other solutions can be constructed by appropriate choice of the coordinates $\lr{x^1,x^2,x^3}$.
In particular, changing the definition of the third coordinate $x^3$ is often sufficient to produce
solutions without trivial symmetries. For example, replacing $x^3=e^z$ with $x^3=z^2$ gives the Beltrami field without continuous Euclidean symmetries 
\begin{equation}
\bol{w}=-\cos\lr{z^2}\,\nabla\lr{e^x\cos y}+\sin\lr{z^2}\nabla\lr{e^x\sin y}.\label{wLE}
\end{equation}
Indeed, projecting equation \eqref{LE} along $\nabla x$ and substituting \eqref{wLE}, we obtain
\begin{equation}\label{LE2}
\begin{split}
\nabla x&\cdot\mf{L}_{\lr{\bol{a}+\bol{b}\cp\bol{x}}}\bol{w}\\
=&\left[\lr{\bol{a}+\bol{b}\cp\bol{x}}\cdot\nabla\right]\left[e^x\lr{\sin\lr{z^2}\sin y-\cos\lr{z^2}\cos y}\right]\\&+b_z e^x\lr{\cos\lr{z^2}\sin y+\sin\lr{z^2}\cos y}\\
=&-\left[\lr{\bol{a}+\bol{b}\cp\bol{x}}\cdot\nabla\right]\left[{e^x\cos\lr{y+z^2}}\right]+b_z e^x\sin\lr{y+z^2}\\
=&e^x\sin\lr{y+z^2}\left[a_y+2z a_z+\lr{2y-1}z b_x-2zx b_y+\lr{1+x}b_z\right]\\
&-e^x\cos\lr{y+z^2}\lr{a_x+zb_y-yb_z}.
\end{split}
\end{equation} 
The vector field \eqref{wLE} possesses a continuous Euclidean symmetry if there exist a vector field $\bol{\xi}_{E}=\bol{a}+\bol{b}\cp\bol{x}\neq\bol{0}$ such that 
the right-hand side of the equation above vanishes in some domain $\Omega\subset\mathbb{R}^3$. 
However, considering the surface $y=-z^2$, it is clear that for the right-hand side of \eqref{LE2} to be zero we must have $a_x=b_y=b_z=0$.
Then, the remaining sine term leads to $a_y=a_z=b_x=0$. Hence, $\bol{a}=\bol{b}=\bol{\xi}_{E}=\bol{0}$.
\end{proof}

\section{Magnetofluidostatic fields with non-vanishing pressure gradients and without Euclidean symmetries}

The purpose of the present section is to show the existence of smooth solutions of system \eqref{MFF}
without continuous Euclidean symmetries and with finite pressure gradients in bounded domains. 

As in the case of Beltrami-type solutions discussed in section 3, the idea is to find a suitable local parametrization
of solutions without trivial symmetries and with non-vanishing pressure gradients. 
Given a domain $\Omega$, a smooth vector field $\bol{w}\in C^{\infty}\lr{\Omega}$ with $\cu\bol{w}\neq\bol{0}$ can be locally represented through the Clebsch parametrization
\begin{equation}
\bol{w}=\nabla\Phi+\Psi\,\nabla\Theta~~~~{\rm in}~~U,\label{w_4}
\end{equation}
with $\lr{\Phi,\Psi,\Theta}\in C^{\infty}\lr{U}$, and $U$ a neighborhood of a point $\bol{x}\in\Omega$ (see e.g. \cite{SatoB1}). 
The vector field $\bol{w}$ is solenoidal provided that
\begin{equation}
\di\bol{w}=\Delta\Phi+\nabla\Psi\cdot\nabla\Theta+\Psi\Delta\Theta=0~~~~{\rm in}~~U.\label{divw}
\end{equation}
Let $\mf{H}\lr{U}$ denote the set of harmonic functions in $U$. 
Suppose that $\nabla\Psi\cdot\nabla\Theta=0$. Since $\cu\bol{w}=\nabla\Psi\cp\nabla\Theta\neq\bol{0}$ in $U$,
$\nabla\Psi\neq\bol{0}$ in $U$. Then, $\Psi\neq 0$ almost everywhere in $U$. From equation \eqref{divw}, this implies $\Phi,\Theta\in\mf{H}\lr{U}$. 
Hence, we look for solutions in the form of \eqref{w_4} such that 
$\Phi,\Theta\in\mf{H}\lr{U}$ and $\nabla\Psi\cdot\nabla\Theta=0$. 
On the other hand, the first equation of system \eqref{MFF} reads as
\begin{equation}
\BN{w}=\left[\nabla\Phi\cdot\nabla\Theta+\Psi\abs{\nabla\Theta}^2\right]\nabla\Psi-\lr{\nabla\Phi\cdot\nabla\Psi}\nabla\Theta=\nabla\chi~~~~{\rm in}~~U.
\end{equation}
Notice that $\cu\bol{w}\cdot\nabla\chi=0$ implies $\chi=\chi\lr{\Psi,\Theta}$. 
Next, perform the change of variables $\Psi=e^\psi$:
\begin{equation}
e^\psi\left[\nabla\Phi\cdot\nabla\Theta+e^\psi\abs{\nabla\Theta}^2\right]\nabla\psi-e^\psi\lr{\nabla\Phi\cdot\nabla\psi}\nabla\Theta=\nabla\chi~~~~{\rm in}~~U.\label{AEq}
\end{equation}
We now consider a Cartesian geometry and demand that $\Theta=x$ (notice that $\Delta x=0$). Equation \eqref{AEq} reduces to
\begin{equation}
e^\psi\lr{\Phi_x+e^{\psi}}\nabla \psi-e^{\psi}\lr{\nabla\Phi\cdot\nabla\psi}\nabla x=\nabla \chi~~~~{\rm in}~~U.\label{AEq2}
\end{equation}
Set 
\begin{equation}
\chi=e^{\psi}\lr{x+\frac{e^\psi}{2}}.\label{chi}
\end{equation}
We have
\begin{equation}
\nabla\chi=e^\psi\lr{x+e^\psi}\nabla\psi+e^\psi\nabla x.\label{dchi}
\end{equation}
Hence, comparing the left-hand side of equation \eqref{AEq2} with the right-hand side of equation \eqref{dchi},
we see that a solution in the form
\begin{equation}
\bol{w}=\nabla\Phi+e^\psi \,\nabla x,
\end{equation}
can be obtained if we can find functions $\Phi$ and $\psi$ such that, in $U$,
\begin{subequations}
\begin{align}
\Delta\Phi&=0,\\
\psi&=\psi\lr{y,z},\\
\Phi_x&=x,\\
\nabla\Phi\cdot\nabla\psi&=-1.
\end{align}
\end{subequations}
Here, the second equations follows from the condition $\nabla\Psi\cdot\nabla\Theta=0$.
Integration of the third equation gives
\begin{equation}
\Phi=\frac{x^2-y^2}{2}+\phi,
\end{equation}
with $\phi\in\mf{H}\lr{U}$ an harmonic function of the $\lr{y,z}$ variables, i.e. $\phi=\phi\lr{y,z}$ and $\Delta\phi=0$.
Now the solution has the form
\begin{equation}
\bol{w}=\nabla\lr{\frac{x^2-y^2}{2}+\phi}+e^{\psi}\,\nabla x,\label{AEq4}
\end{equation}
and the conditions on $\phi$ and $\psi$ become
\begin{subequations}\label{AEq5}
\begin{align}
\Delta\phi&=0,\label{Lapphi}\\
\psi&=\psi\lr{y,z},\\
\phi&=\phi\lr{y,z},\\
-y\psi_y+\nabla\phi\cdot\nabla\psi&=-1.\label{psiphi}
\end{align}
\end{subequations}
Observe that, once $\phi$ is given as a solution of \eqref{Lapphi}, the 
function $\psi$ can be obtained through the method of characteristics from equation \eqref{psiphi}. 
Below, explicit solutions of system \eqref{MFF} in the form \eqref{AEq4} are given.

\begin{example}\label{ex4_1}
A simple solution of system \eqref{AEq5} is $\phi=-\psi=z$. The corresponding solution of system \eqref{MFF} is
\begin{equation}
\bol{w}=\nabla\lr{\frac{x^2-y^2}{2}+z}+e^{-z}\,\nabla x,~~~~\chi=e^{-z}\lr{x+\frac{e^{-z}}{2}}.\label{w4_1}
\end{equation} 
\end{example}

\begin{example}\label{ex4_2}
The functions $\phi=z$ and $\psi=z+2\log y$ are solutions of system \eqref{AEq5}.
The corresponding solution of system \eqref{MFF} is
\begin{equation}
\bol{w}=\nabla\lr{\frac{x^2-y^2}{2}+z}+y^2e^{z}\,\nabla x,~~~~\chi=y^2e^{z}\lr{x+\frac{y^2 e^{z}}{2}}.\label{w4_2}
\end{equation} 
\end{example}

\begin{example}\label{ex4_3}
The functions $\phi=\frac{z^2-y^2}{2}$ and $\psi=\log \lr{yz}$ are solutions of system \eqref{AEq5}.
The corresponding solution of system \eqref{MFF} is
\begin{equation}
\bol{w}=\nabla\lr{\frac{x^2+z^2}{2}-y^2}+yz\,\nabla x,~~~~\chi=yz\lr{x+\frac{yz}{2}}.\label{w4_3}
\end{equation} 
\end{example}

\begin{example}\label{ex4_4}
More generally, the functions 
\begin{equation}
\begin{split}
\phi&=\alpha\lr{\frac{z^2-y^2}{2}}+\beta z+\gamma y,\\
\psi&=\frac{1}{1+\alpha}\log\left[\lr{1+\alpha}y-\gamma\right]+\delta\left[\gamma-\lr{1+\alpha}y\right]^{\frac{\alpha}{1+\alpha}}\lr{\frac{\beta}{\alpha}+z},
\end{split}
\end{equation}
are solutions of system \eqref{AEq5}.
Here, the real constants $\alpha,\beta,\gamma,\delta\in\mathbb{R}$ are chosen so that $\psi$ is real and smooth  
in the domain of interest. 
The corresponding solution of system \eqref{MFF} is
\begin{equation}
\begin{split}
\bol{w}=&\nabla\lr{\frac{x^2+\alpha z^2-\lr{1+\alpha}y^2}{2}+\beta z+\gamma y}\\
&+\left[\lr{1+\alpha}y-\gamma\right]^{\frac{1}{1+\alpha}}\exp\left\{\delta\left[\gamma-\lr{1+\alpha}y\right]^{\frac{\alpha}{1+\alpha}}\lr{\frac{\beta}{\alpha}+z}\right\}\nabla x.
\end{split}
\end{equation} 
The long expression for the pressure field $\chi$ was omitted since it can be obtained 
easily from equation \eqref{chi}.
\end{example}

Let us examine the symmetry properties of these solutions.
We have the following:

\begin{proposition}\label{prop6}
There exist solutions of system \eqref{MFF} with non-vanishing pressure gradients and without continuous
Euclidean symmetries, i.e. smooth solenoidal vector fields $\bol{w}\in C^{\infty}\lr{\Omega}$ 
solving \eqref{MFF} in some bounded domain $\Omega\subset\mathbb{R}^3$ with $\nabla\chi\neq\bol{0}$ 
and satisfying the condition \eqref{SE3asym} in $\Omega$.  
\end{proposition}

\begin{proof}
Consider the solution \eqref{w4_1} given in example \ref{ex4_1}. 
The vector field $\bol{w}$ in equation \eqref{w4_1} 
can be rearranged as follows:
\begin{equation}
\bol{w}=\lr{x+e^{-z}}\,\nabla x-y\,\nabla y+\nabla z.
\end{equation}
Projecting equation \eqref{LE} along $\nabla x$, we have:
\begin{equation}
\begin{split}
\nabla x\cdot\mf{L}_{\lr{\bol{a}+\bol{b}\cp\bol{x}}}\bol{w}=&
a_x+b_y \lr{z-1}-2b_z y-e^{-z}\lr{a_z+b_x y-b_y x}.\label{p6_Lx}
\end{split}
\end{equation}
We choose the domain $\Omega$ to be the unit open ball $B\lr{\bol{0}}$ in $\mathbb{R}^3$ centered at $\bol{x}=\bol{0}$. 
On the $z$-axis (corresponding to $x=y=0$), the expression 
on the right-hand side of \eqref{p6_Lx} vanishes provided that
\begin{equation}
a_x+b_y\lr{z-1}-a_z e^{-z}=0.
\end{equation}
Hence, $a_x=a_z=b_y=0$ and the condition that the $x$-component of the Lie-derivative \eqref{p6_Lx} is zero reduces to
\begin{equation}
2b_z+b_x e^{-z}=0.\label{p6_Lx2}
\end{equation}
It follows that $b_x=b_z=0$ so that $\bol{a}+\bol{b}\cp\bol{x}=a_y\p_y$. 
Next, consider the projection of \eqref{LE} along $\nabla y$:
\begin{equation}
\nabla y\cdot\mf{L}_{\lr{a_y \p_y}}\bol{w}=-a_y.
\end{equation}
The right-hand side vanishes if and only if $a_y=0$.
We have thus shown that $\bol{a}+\bol{b}\cp\bol{x}=\bol{0}$.
This implies that the smooth solenoidal vector field \eqref{w4_1} in the unit open ball does not possess continuous Euclidean symmetries 
while exhibiting a non-vanishing pressure gradient
\begin{equation}
\nabla\chi=-e^{-z}\lr{x+e^{-z}}\nabla z+e^{-z}\nabla x.
\end{equation}
This proves proposition \ref{prop6}. 

For completeness, let us verify that also the examples
\eqref{w4_2} and \eqref{w4_3} do not possess continuous Euclidean symmetries. 
The vector field $\bol{w}$ in \eqref{w4_2} can be rearranged as
\begin{equation}
\bol{w}=\lr{x+y^2 e^z}\,\nabla x-y\,\nabla y+\nabla z.
\end{equation}
Projecting equation \eqref{LE} along $\nabla x$, we have:
\begin{equation}
\begin{split}
\nabla x\cdot\mf{L}_{\lr{\bol{a}+\bol{b}\cp\bol{x}}}\bol{w}=&
a_x+b_y \lr{z-1}-2 b_z y+2y e^z\lr{a_y+b_z x-b_x z}\\&+y^2 e^z\lr{a_z+b_x y-b_y x}.\label{p62_Lx}
\end{split}
\end{equation}
Again, we set $\Omega=B\lr{\bol{0}}$. On the plane $y=0$, we have the condition
\begin{equation}
a_x+b_y\lr{z-1}=0.
\end{equation}
This implies $a_x=b_y=0$. 
Next, consider the $y$-axis $x=z=0$:
\begin{equation}
2\lr{a_y-b_z}+a_z y+b_x y^2=0.
\end{equation}
Therefore, $a_z=b_x=0$ and $a_y=b_z$.
Equation \eqref{p62_Lx} now reads as
\begin{equation}
\begin{split}
\nabla x\cdot\mf{L}_{\lr{\bol{a}+\bol{b}\cp\bol{x}}}\bol{w}=&
2 a_y y\left[ e^z\lr{1+ x}- 1\right].\label{p62_Lx}
\end{split}
\end{equation}
We thus conclude that $a_y=b_z=0$, and the vector field $\bol{w}$ of 
equation \eqref{w4_2}
does not exhibit continuous Euclidean symmetries. 

The vector field $\bol{w}$ in \eqref{w4_3} can be rearranged as
\begin{equation}
\bol{w}=\lr{x+yz}\,\nabla x-2y\,\nabla y+z\nabla z.
\end{equation}
Projecting equation \eqref{LE} along $\nabla x$, we have:
\begin{equation}
\begin{split}
\nabla x\cdot\mf{L}_{\lr{\bol{a}+\bol{b}\cp\bol{x}}}\bol{w}=&
a_x-3b_z y+z\lr{a_y+b_z x-b_x z}\\&+y\lr{a_z+b_x y-b_y x}.
\label{p63_Lx}
\end{split}
\end{equation}
Set $\Omega=B\lr{\bol{0}}$. On the $x$-axis $y=z=0$, we have the condition
$a_x=0$. Then, on the $y$-axis $z=x=0$, 
\begin{equation}
-3b_z +a_z+b_x y=0.
\end{equation}
This implies $b_x=0$ and $a_z=3b_z$.
Next, consider the $z$-axis $x=y=0$. We have $a_y=0$. 
The condition that the righ-hand side of equation \eqref{p63_Lx} vanishes now reads as
\begin{equation}
b_z z -b_y  y=0.
\end{equation}
It follows that $b_z=b_y=a_z=0$. We conclude that the vector field $\bol{w}$ in \eqref{w4_3} does not possess continuous Euclidean symmetries. 
\end{proof}

Proposition \eqref{prop6} can now be used to construct square integrable solutions $\bol{w}\in L^2\lr{\Omega}$ of the boundary value problem
\eqref{MFF}, \eqref{BC} without continuous Euclidean symmetries and
with non-vanishing pressure gradients:

\begin{theorem}\label{thm1}
Let $\Omega\subset\mathbb{R}^3$ denote a smoothly bounded domain with boundary $\p\Omega$. Let $\bol{n}$ denote the unit outward normal to $\p\Omega$. Then, the boundary value problem
\begin{equation}\label{TChi}
\begin{split}
&\BN{w}=\nabla\chi,~~~~\di\bol{w}=0~~~~{\rm in}~~\Omega,\\
&\bol{w}\cdot\bol{n}=0~~~~{\rm on}~~\p\Omega,
\end{split}
\end{equation}  
admits solutions $\bol{w}\in L^2\lr{\Omega}$ such that
there exist an open set $U\subset\Omega$ with
\begin{equation}
\begin{split}
&\bol{w}\in C^{\infty}\lr{U},\\
&\nabla\chi\neq\bol{0}~~~~{\rm in}~~U,\\
&\mf{L}_{\lr{\bol{a}+\bol{b}\cp\bol{x}}}\bol{w}\neq\bol{0}~~\forall\bol{a},\bol{b}\in\mathbb{R}^3:\bol{a}+\bol{b}\cp\bol{x}\neq\bol{0}~~~~{\rm in}~~U.
\end{split}
\end{equation}
\end{theorem}

\begin{proof}
First, choose the origin of a Cartesian coordinate system to be a point $p\in\Omega$. Let $B\lr{p,\epsilon}$ be an open ball centered at $p$ with radius $\epsilon>0$ contained in $\Omega$, $\ov{B}\lr{p,\epsilon}\subset\Omega$. Here, the overbar denotes the closure of a set.
We define $A=\Omega-\ov{B}\lr{p,\epsilon}$ and consider the boundary value problem
\begin{equation}\label{TBel}
\begin{split}
&\cu\bol{w}=\hat{h}\,\bol{w}~~~~{\rm in}~~A,\\
&\bol{w}\cdot\bol{n}_A=0~~~~{\rm on}~~\p A.
\end{split}
\end{equation} 
Here $\hat{h}\in\mathbb{R}$ is a real constant, $\p A$ the boundary of $A$, and $\bol{n}_A$ the unit outward normal to $\p A$. 
Observe that, by construction, the domain $A$ is multiply connected
and smoothly bounded. Hence, from theorem 2 of \cite{YRot}, the boundary value problem \eqref{TBel} admits a strong Beltrami field solution $\bol{w}_{\rm Bel}$. 
In particular, $\bol{w}_{\rm Bel}\in H^1\lr{A}$, with $H^1$ the standard Sobolev
space of order $1$. 
The desired solution of \eqref{TChi} can be given as
\begin{equation}\label{SolBVP}
\bol{w}=\begin{cases}
\bol{w}_{\rm Bel}~~~~{\rm if}~~\bol{x}\in \ov{A}\\
\nabla\lr{\frac{x^2-y^2}{2}+z}+e^{-z}\nabla x~~~~{\rm if}~~\bol{x}\in B\lr{p,\epsilon}
\end{cases}.
\end{equation}
Evidently $\bol{w}\in L^2\lr{\Omega}$ and $\nabla\chi\neq\bol{0}$ in $B\lr{p,\epsilon}$. Furthermore, 
since \eqref{SolBVP} coincides with \eqref{w4_1} in $B\lr{p,\epsilon}$, 
this solution does not possess continuous Euclidean symmetries 
as a consequence of proposition \ref{prop6}. 
\end{proof}

\section{Symmetries of the Beltrami equation}

In the previous sections we have shown that there exist ideal magnetofluid equilibria without continuous Euclidean symmetries. 
However, these solutions were constructed ad hoc by carefully
calibrating the geometry of local coordinates and Clebsch parameters.   
Therefore, it is desirable to have a systematic method to obtain new solutions with given symmetry properties. The purpose of the present section is to study the symmetries of the Beltrami equation \eqref{BF}. Here, a symmetry of a differential equation is intended as a transformation mapping solutions into solutions. We will see that these symmetries, which should not be confused with the symmetries of the solutions, enable us to derive
new solutions without Euclidean symmetries from known ones. 
In particular, it will be shown that, if the proportionality coefficient $\hat{h}$ of a solenoidal Beltrami field $\bol{w}$ exhibits a continuous Euclidean symmetry $\bol{\xi}_E$ such that $\mf{L}_{\bol{\xi}_E}\hat{h}=0$,  
then, the transformed vector field $\mf{L}_{\bol{\xi}_E}\bol{w}$ is either the null vector or a solenoidal Beltrami field with the same proportionality coefficient $\hat{h}$. 

Let $\mf{B}_{\hat{h}}\lr{\Omega}$ denote the set comprising solenoidal solutions of equation \eqref{BF} in a bounded domain $\Omega$ with given proportionality coefficient $\hat{h}\in C^{\infty}\lr{\Omega}$ and the null vector $\bol{0}$:
\begin{equation}  
\mf{B}_{\hat{h}}\lr{\Omega}=\left\{\bol{w}\in C^{\infty}\lr{\Omega}~\rvert~\cu\bol{w}=\hat{h}\,\bol{w},~\di\bol{w}=0\right\}\cup\left\{\bol{0}\right\}. 
\end{equation}
%Note that $\mf{B}_{\hat{h}}\lr{\Omega}$ includes the null vector $\bol{0}$. 
The following holds:

\begin{proposition}\label{prop7}
Let $\bol{\xi}_{E}=\bol{a}+\bol{b}\cp\bol{x}$, $\bol{a},\bol{b}\in\mathbb{R}^3$, denote a continuous Euclidean symmetry.
Let $\bol{w}\in T\Omega$ be a solenoidal vector field, in general not symmetric with respect to $\bol{\xi}_{E}$. 
Then, the Lie-derivative and the curl operator commute, 
\begin{equation}
\mf{L}_{\bol{\xi}_{E}}\cu\bol{w}=\cu\mf{L}_{\bol{\xi}_{E}}\bol{w}.
\end{equation}
\end{proposition}

\begin{proof}
Let $\bol{\xi}\in T\Omega$ denote a solenoidal vector field. We have
%\begin{equation}
%\begin{split}
%\mf{L}_{\bol{\xi}}\cu\bol{w}&=\nabla\cp\left[\lr{\cu\bol{w}}\cp\bol{\xi}\right]=\nabla\hat{h}\cp\lr{\bol{w}\cp\bol{\xi}}+\hat{h}\cu\lr{\bol{w}\cp\bol{\xi}}\\
%&=(\nabla\hat{h}\cdot\bol{\xi})\bol{w}+\hat{h}\mf{L}_{\bol{\xi}}\bol{w}.
%\end{split}
%\end{equation}
%Here, we used the fact that $\nabla\hat{h}\cdot\bol{w}=0$. On the other hand,
\begin{equation}
\begin{split}
\cu\mf{L}_{\bol{\xi}}\bol{w}&=\cu\left[\cu\lr{\bol{w}\cp\bol{\xi}}\right]\\
&=\cu\left[\lr{\bol{\xi}\cdot\nabla}\bol{w}-\lr{\bol{w}\cdot\nabla}\bol{\xi}\right]\\
&=\cu\left[-2\lr{\bol{w}\cdot\nabla}\bol{\xi}+\nabla\lr{\bol{\xi}\cdot\bol{w}}-\bol{w}\cp\lr{\cu\bol{\xi}}-\bol{\xi}\cp\cu\bol{w}\right]\\
&=\mf{L}_{\bol{\xi}}\cu\bol{w}+\cu\left[\lr{\cu\bol{\xi}}\cp\bol{w}-2\lr{\bol{w}\cdot\nabla}\bol{\xi}\right].
\end{split}
\end{equation}
If $\bol{\xi}=\bol{\xi}_{E}$, we have
\begin{equation}
\lr{\cu\bol{\xi}_{E}}\cp\bol{w}-2\lr{\bol{w}\cdot\nabla}\bol{\xi}_{E}=2\bol{b}\cp\bol{w}-2\bol{b}\cp\bol{w}=\bol{0}.
\end{equation}
It follows that
\begin{equation}
\mf{L}_{\bol{\xi}_{E}}\cu\bol{w}=\cu\mf{L}_{\bol{\xi}_{E}}\bol{w}.
\end{equation}
\end{proof}

\begin{proposition}\label{prop8}
Let $\bol{w}\in\mf{B}_{\hat{h}}\lr{\Omega}$ denote a solenoidal Beltrami field  with proportionality coefficient $\hat{h}$ in a bounded domain $\Omega$. 
Suppose that $\mf{L}_{\bol{\xi}_{E}}\hat{h}=0$ for some continuous Euclidean symmetry $\bol{\xi}_{E}=\bol{a}+\bol{b}\cp\bol{x}$, $\bol{a},\bol{b}\in\mathbb{R}^3$.
Then, $\mf{L}_{\bol{\xi}_{E}}^n\bol{w}\in\mf{B}_{\hat{h}}\lr{\Omega}$, that is $\bol{w}$ is either the null vector or a solenoidal Beltrami field with the same proportionality coefficient for all $n\in\mathbb{N}$: 
\begin{equation}
\cu\mf{L}_{\bol{\xi}_{E}}^n\bol{w}=\cu\mf{L}_{\bol{\xi}_{E}}...\mf{L}_{\bol{\xi}_{E}}\bol{w}=\hat{h}\,\mf{L}_{\bol{\xi}_{E}}^n\bol{w}~~~~\forall n\in\mathbb{N}.
\end{equation}
\end{proposition}

\begin{proof}
Set $n=1$. We have
\begin{equation}
\begin{split}
\mf{L}_{\bol{\xi}_{E}}\cu\bol{w}&=\nabla\cp\left[\lr{\cu\bol{w}}\cp\bol{\xi}_{E}\right]=\nabla\hat{h}\cp\lr{\bol{w}\cp\bol{\xi}_{E}}+\hat{h}\cu\lr{\bol{w}\cp\bol{\xi}_{E}}\\
&=(\nabla\hat{h}\cdot\bol{\xi}_{E})\bol{w}+\hat{h}\mf{L}_{\bol{\xi}_{E}}\bol{w}=(\mf{L}_{\bol{\xi}_{E}}\hat{h})\bol{w}+\hat{h}\mf{L}_{\bol{\xi}_{E}}\bol{w}=\hat{h}\mf{L}_{\bol{\xi}_{E}}\bol{w}.
\end{split}
\end{equation}
Here, we used the fact that $\bol{w}\cdot\nabla\hat{h}=0$. 
From proposition \ref{prop7} the curl operator and Lie-derivative commute. Hence,
\begin{equation}
\cu\mf{L}_{\bol{\xi}_{E}}\bol{w}=\hat{h}\mf{L}_{\bol{\xi}_{E}}\bol{w}.
\end{equation}
The procedure can be repeated for all $n>1$ by setting $\bol{w}^{n-1}=\mf{L}_{\bol{\xi}_{E}}\bol{w}^{n-2}=\mf{L}_{\bol{\xi}_{E}}^{n-1}\bol{w}^0$ with $\bol{w}^0=\bol{w}$. 
Then, 
\begin{equation}
\cu\mf{L}_{\bol{\xi}_{E}}\bol{w}^{n-1}=\hat{h}\,\mf{L}_{\bol{\xi}_{E}}\bol{w}^{n-1}.
\end{equation}
This implies
\begin{equation}
\cu\mf{L}_{\bol{\xi}_{E}}^n\bol{w}=\hat{h}\,\mf{L}_{\bol{\xi}_{E}}^n\bol{w}~~~~\forall~~n\in\mathbb{N}.
\end{equation}
Finally, observe that, if $\bol{w}$ is a solenoidal vector field, the Lie derivative $\mf{L}_{\bol{\xi_E}}\bol{w}=\cu\lr{\bol{w}\cp\bol{\xi}_E}$ is itself solenoidal. Hence $\mf{L}^n_{\bol{\xi}_E}\bol{w}$ is solenoidal.
\end{proof}

\begin{remark}
Notice that the procedure described in proposition \ref{prop8} may end when $\mf{L}_{\bol{\xi}_{E}}^m\bol{w}=\bol{0}$ for some $m\in\mathbb{N}$.
Then, all successive elements of the sequence $\mf{L}_{\bol{\xi}_{E}}^n\bol{w}$, $n>m$, are trivial Beltrami fields $\mf{L}_{\bol{\xi}_{E}}^n\bol{w}=\bol{0}$. 
\end{remark}

\begin{remark}
Let $\bol{w}\in\mf{B}_{\hat{h}}\lr{\Omega}$ be a solenoidal Beltrami field.
Let $\bol{w}'$ denote the result of the transformation of $\bol{w}$ according to the Euclidean isometry associated with $\bol{\xi}_E=\bol{a}+\bol{b}\cp\bol{x}$. Let $\epsilon>0$ be a small constant parametrizing the transformation. We have
\begin{equation}
\bol{w}'=\bol{w}+\epsilon\,\mf{L}_{\bol{\xi}_E}\bol{w}+o\lr{\epsilon^2}.
\end{equation}
Proposition \ref{prop8} can be interpreted as follows:
if the proportionality coefficient of a solenoidal Beltrami field
has a continuous Euclidean symmetry $\bol{\xi}_E$, then 
the infinitesimal transformation of the field is itself a 
solenoidal Beltrami field with the same proportionality coefficient. 
Hence, the transformation associated with $\bol{\xi}_E$ defines
a symmetry of the Beltrami equation. 
\end{remark}

\noindent The following is an example of application of proposition \ref{prop8}.

\begin{example}
Consider the solenoidal Beltrami field encountered in proposition \ref{prop5}
\begin{equation}
\begin{split}
\bol{w}&=-\cos\lr{z^2}\,\nabla\lr{e^x\cos y}+\sin\lr{z^2}\,\nabla\lr{e^x\sin y}\\
&=e^x\left[-\cos\lr{y+z^2}\nabla x+\sin\lr{y+z^2}\nabla y\right].
\end{split}
\end{equation}
Notice that $\hat{h}=2z$. Set $\bol{\xi}_{E}=\nabla x$. We have $\mf{L}_{\nabla x}\hat{h}=2\nabla x\cdot\nabla z=0$. Hence,
according to proposition \ref{prop8}, the vector field
\begin{equation}
\mf{L}_{\nabla x}\bol{w},
\end{equation}
is itself a solenoidal Beltrami field with proportionality coefficient $\hat{h}=2z$. We have
\begin{equation}
\mf{L}_{\nabla x}\bol{w}=\cu\lr{\bol{w}\cp\nabla x}=-\cu\left\{e^x\left[{\sin y \cos\lr{z^2}+\cos y \sin\lr{z^2}}\right]\nabla z\right\}=\bol{w}.
\end{equation}
Next, set $\bol{a}=\bol{0}$ and $\bol{b}=\nabla z$, so that $\bol{\xi}_{E}=x\nabla y-y\nabla x$ and $\mf{L}_{\lr{x\nabla y-y\nabla x}}\hat{h}=0$. 
We have
\begin{equation}
\begin{split}
\mf{L}_{\lr{x\nabla y-y\nabla x}}\bol{w}=&\cu\left[{\bol{w}\cp\lr{x\nabla y-y\nabla x}}\right]\\
=&e^x\left[\lr{1+x}\sin\lr{y+z^2}+y\cos\lr{y+z^2}\right]\nabla x\\&+e^x\left[\lr{1+x}\cos\lr{y+z^2}-y\sin\lr{y+z^2}\right]\nabla y.
\end{split}
\end{equation}
One can verify that $\mf{L}_{\lr{x\nabla y-y\nabla x}}\bol{w}$ is a solenoidal Beltrami field with proportionality coefficient $\hat{h}=2z$ and without
continuous Euclidean symmetries. 
 
\end{example}

\section{Concluding Remarks}

The existence of magnetofluidostatic fields without continuous Euclidean symmetries
is a mathematical problem encountered in the design of plasma confinement devices (stellarators)
and in the development of nuclear fusion reactors. 
In this paper, the existence of ideal magnetofluid equilibria without
continuous Euclidean symmetries (combinations of translations and rotations) was studied.

First, we showed that all magnetofluidostatic fields are locally symmetric,
in the sense that one can find a solenoidal vector field defined locally  
with the property that the Lie derivative of the solution with respect to the local field is identically zero.
However, this symmetry is not necessarily an Euclidean isometry.
In particular, there exists smooth magnetofluidostatic fields defined in bounded domains
that do not possess any continuous Euclidean symmetry.
This is true for both Beltrami-type solutions (proposition \ref{prop5}), and non-vanishing pressure gradients solutions (proposition \ref{prop6}).
By combining Beltrami-type solutions with non-vanishing pressure gradient solutions,
it is also possible to construct square integrable solutions of the boundary value
problem, as proven in theorem \ref{thm1}.
Finally, the symmetry properties of the Beltrami equation were studied.
We found that, if the proportionality coefficient of a Beltrami field
possesses a continuous Euclidean symmetry, than the associated transformation
defines a symmetry of the Beltrami equation.
Hence, new solutions can be computed from known ones by 
application of the Lie derivative (proposition \ref{prop8}).

%variational formulation!!
% Refer to Hudson properly, read remaining refs, add refs
%increase refs, repeat Weitzner calculations
%check all formulas including those in the text (in particular introduction)

\section*{Acknowledgments}

\noindent The research of N. S. was supported by JSPS KAKENHI Grant No. 18J01729.
          The author is grateful to Professor Z. Yoshida for useful discussion on $MHD$ equilibria and the notion of symmetry, 
					to Dr. Z. Qu and Professor R. L. Dewar for useful discussion on magnetostatics, and to Professor M. Yamada
					for useful discussion on fluid systems.

\appendix

\section{Derivation of the Grad-Shafranov equation}

Let $\lr{x^1,x^2,x^3}$ denote a coordinate system. 
Suppose that the solution $\bol{w}$ of \eqref{MFF} has an ignorable coordinate $x^3$:
\begin{equation}
\frac{\p w^i}{\p x^3}=0,~~~~\frac{\p\chi}{\p x^3}=0,~~~~\frac{\p g_{ij}}{\p x^3}=0,~~~~i,j=1,2,3.\label{A1}
\end{equation}
As usual, $g_{ij}$ denotes a component of the covariant metric tensor.
It is not difficult to verify that the conditions above are equivalent to demanding that
\begin{equation}
\mf{L}_{\p_3}\bol{w}=\bol{0},~~~~\mf{L}_{\p_3}\chi=0,~~~~\mf{L}_{\p_3}g_{ij}dx^i\otimes dx^j=0. 
\end{equation}
Furthermore, as shown in proposition \ref{prop3}, $\p_3$ must be the combination of a rotation followed by a translation, i.e. $\p_3=\bol{a}+\bol{b}\cp\bol{x}$ for some $\bol{a},\bol{b}\in\mathbb{R}^3$.  
Next, observe that the first equation of system \eqref{MFF} can be written as: 
\begin{equation}
\begin{split}
\BN{w}&=w^i \p_{i}\cp\nabla\cp\lr{w_{j}\nabla x^j}=w^i\p_i \cp\lr{\nabla w_{j}\cp\nabla x^j}\\
&=w^i\left[\delta_{i}^{\,j}\nabla w_{j}-\lr{\p_{i}\cdot\nabla w_{j}}\nabla x^j\right]=w^i\lr{\frac{\p w_{i}}{\p x^j}-\frac{\p w_j}{\p x^i}}\nabla x^j\\
&=\chi_1\nabla x^1+\chi_2\nabla x^2.\label{wxdw}
\end{split}
\end{equation}
Substituting the first two equations of system \eqref{A1} in system \eqref{wxdw}, we obtain
\begin{subequations}
\begin{align}
&w^2\lr{\frac{\p w_{2}}{\p x^1}-\frac{\p w_1}{\p x^2}}+w^3\frac{\p w_{3}}{\p x^1}=\chi_1,\label{j1}\\
&w^1\lr{\frac{\p w_1}{\p x^2}-\frac{\p w_{2}}{\p x^1}}+w^3\frac{\p w_{3}}{\p x^2}=\chi_2,\label{j2}\\
&-w^1\frac{\p w_{3}}{\p x^1}-w^2\frac{\p w_{3}}{\p x^2}=0.\label{j3}
\end{align}
\end{subequations}
%Since $J^{-1}=1$, we also have $\mf{L}_{\p_{3}}J=0$.
Let $J=\p_1\cdot\p_2\cp\p_3$ be the Jacobian of the transformation
$\lr{x,y,z}\mapsto\lr{x^1,x^2,x^3}$.  
Note that here  
$w^2$ denotes the component $w^2=\bol{w}\cdot\nabla x^2$, and should not be confused with the square of the norm of $\bol{w}$. 
Then, the divergence-free condition $\di\bol{w}=0$ reads as
\begin{equation}
\frac{\p}{\p x^1}\lr{Jw^1}+\frac{\p}{\p x^2}\lr{Jw^2}=0.\label{div0}
\end{equation} 
As in the proof of proposition \ref{prop2}, define the stream function $\Theta\lr{x^1,x^2}$ as
\begin{equation}
\Theta\lr{x^1,x^2}=-\int_{a}^{x^1}{J w^2}\lr{s^1,x^2}\,ds^1+\int_{b}^{x^2}{Jw^1}\lr{a,s^2}\,ds^2,
\end{equation}
where $\lr{a,b,c}$ denotes a point in the domain of interest. 
Using equation \eqref{div0}, 
the components $w^1$ and $w^2$ can now be expressed as
\begin{equation}
w^1=\frac{1}{J}\frac{\p\Theta}{\p x^2},~~~~w^2=-\frac{1}{J}\frac{\p\Theta}{\p x^1}.\label{w12}
\end{equation}
Substituting \eqref{w12} in \eqref{j3} we obtain
\begin{equation}
\frac{\p \Theta}{\p x^2}\frac{\p w_3}{\p x^1}-\frac{\p \Theta}{\p x^1}\frac{\p w_3}{\p x^2}=0.
\end{equation}
This implies $\nabla\Theta\cp\nabla w_3=\bol{0}$, so that $w_3=w_3\lr{\Theta}$ depends only on the stream function $\Theta$.
We have thus shown that
\begin{equation}
\bol{w}=\nabla\Theta\cp\nabla x^3+w^{3}\,\p_{3}.
\end{equation}
In a similar way,
\begin{equation}
\nabla\chi\cdot\bol{w}=J^{-1}\lr{\chi_1\Theta_2-\chi_2\Theta_1}=0,
\end{equation}
implies that $\chi=\chi\lr{\Theta}$. 
Next, combine equations \eqref{j1} and \eqref{j2} as follows:
\begin{equation}
\begin{split}
&\left[w^2\lr{\frac{\p w_{2}}{\p x^1}-\frac{\p w_1}{\p x^2}}+w^3\frac{\p w_{3}}{\p x^1}-\chi_1\right]\nabla x^1\\
&+\left[w^1\lr{\frac{\p w_1}{\p x^2}-\frac{\p w_{2}}{\p x^1}}+w^3\frac{\p w_{3}}{\p x^2}-\chi_2\right]\nabla x^2\\
&=
\frac{1}{J}\lr{\frac{\p w_1}{\p x^2}-\frac{\p w_2}{\p x^1}}\nabla\Theta-\nabla\chi+w^3\nabla w_3=\bol{0}.
\end{split}
\end{equation}
Since $w_{3}$ and $\chi$ are functions of $\Theta$, the equation above can be rewritten as
\begin{equation}
\frac{1}{J}\lr{\frac{\p w_1}{\p x^2}-\frac{\p w_{2}}{\p x^1}}-\frac{d\chi}{d\Theta}+w^3\frac{dw_3}{d\Theta}=0.\label{GGS0}
\end{equation}
Now we need to express $w_{1}$, $w_{2}$ and $w^3$ in terms of $\Theta$ and the metric coefficients. From the identity $w_{i}=g_{ij}w^j$ we have
\begin{subequations}
\begin{align}
w_{1}&=\frac{1}{J}\lr{g_{11}\frac{\p\Theta}{\p x^2}-g_{12}\frac{\p\Theta}{\p x^1}}+g_{13}w^3,\label{w_1}\\
w_{2}&=\frac{1}{J}\lr{g_{21}\frac{\p\Theta}{\p x^2}-g_{22}\frac{\p\Theta}{\p x^1}}+g_{23}w^3,\label{w_2}\\
w_{3}&=\frac{1}{J}\lr{g_{31}\frac{\p\Theta}{\p x^2}-g_{32}\frac{\p\Theta}{\p x^1}}+g_{33}w^3.
\end{align}
\end{subequations}
From the third equation, it follows that
\begin{equation}
w^3=\frac{1}{g_{33}}\left[w_{3}\lr{\Theta}+\frac{1}{J}\lr{g_{32}\frac{\p\Theta}{\p x^1}-g_{31}\frac{\p\Theta}{\p x^2}}\right].\label{w3}
\end{equation}
Substituting this expression in \eqref{w_1} and \eqref{w_2}, we obtain
\begin{equation}\label{w_1w_2}
\begin{split}
w_{1}=\frac{g_{13}}{g_{33}}w_{3}+\frac{1}{J}\left[\frac{\p\Theta}{\p x^1}\lr{\frac{g_{13}g_{32}-g_{12}g_{33}}{g_{33}}}+\frac{\p\Theta}{\p x^2}\lr{\frac{g_{11}g_{33}-g_{13}g_{31}}{g_{33}}}\right],\\
w_{2}=\frac{g_{23}}{g_{33}}w_{3}+\frac{1}{J}\left[\frac{\p\Theta}{\p x^1}\lr{\frac{g_{32}g_{23}-g_{22}g_{33}}{g_{33}}}+\frac{\p\Theta}{\p x^2}\lr{\frac{g_{21}g_{33}-g_{31}g_{23}}{g_{33}}}\right].
\end{split}
\end{equation}
On the other hand, the following identities hold:
\begin{subequations}
\begin{align}
g_{13}g_{32}-g_{12}g_{33}&=\lr{\p_{1}\cdot\p_{3}}\lr{\p_{3}\cdot\p_{2}}-\lr{\p_{1}\cdot\p_{2}}\lr{\p_{3}\cdot\p_{3}}=\p_{1}\cp\p_{3}\cdot\p_{3}\cp\p_{2},\\
g_{11}g_{33}-g_{13}g_{31}&=\lr{\p_{1}\cdot\p_{1}}\lr{\p_{3}\cdot\p_{3}}-\lr{\p_{1}\cdot\p_{3}}\lr{\p_{3}\cdot\p_{1}}=\p_{1}\cp\p_{3}\cdot\p_{1}\cp\p_{3},\\
g_{32}g_{23}-g_{22}g_{33}&=\lr{\p_{3}\cdot\p_{2}}\lr{\p_{2}\cdot\p_{3}}-\lr{\p_{2}\cdot\p_{2}}\lr{\p_{3}\cdot\p_{3}}=\p_{3}\cp\p_{2}\cdot\p_{2}\cp\p_{3},\\
g_{21}g_{33}-g_{31}g_{23}&=\lr{\p_{2}\cdot\p_{1}}\lr{\p_{3}\cdot\p_{3}}-\lr{\p_{3}\cdot\p_{1}}\lr{\p_{2}\cdot\p_{3}}=\p_{2}\cp\p_{3}\cdot\p_{1}\cp\p_{3}.
\end{align}
\end{subequations}
Furthermore, the tangent basis $\p_{i}$ is related to the cotangent basis $dx^{i}$ according to $\nabla x^i=\epsilon^{ijk}J^{-1}\p_j\cp \p_{k}$. Hence, 
\begin{subequations}
\begin{align}
g_{13}g_{32}-g_{12}g_{33}&=J^2\nabla x^1\cdot\nabla x^2=J^2 g^{12},\\
g_{11}g_{33}-g_{13}g_{31}&=J^2\nabla x^2\cdot\nabla x^2=J^2 g^{22},\\
g_{32}g_{23}-g_{22}g_{33}&=-J^2\nabla x^1\cdot\nabla x^1=-J^2 g^{11},\\
g_{21}g_{33}-g_{31}g_{23}&=-J^2\nabla x^1\cdot\nabla x^2=-J^2 g^{12}.
\end{align}
\end{subequations}
Using these expressions, equation \eqref{w_1w_2} becomes
\begin{equation}\label{w_1w_22}
\begin{split}
w_1&=\frac{1}{g_{33}}\left[g_{13}w_3+J\lr{g^{12}\frac{\p\Theta}{\p x^1}+g^{22}\frac{\p \Theta}{\p x^2}}\right],\\
w_2&=\frac{1}{g_{33}}\left[g_{23}w_3-J\lr{g^{11}\frac{\p\Theta}{\p x^1}+g^{12}\frac{\p \Theta}{\p x^2}}\right].
\end{split}
\end{equation}
Plugging equations \eqref{w_1w_22} and \eqref{w3} into \eqref{GGS0}, we arrive at
\begin{equation}
\frac{1}{J}\frac{\p}{\p x^i}\lr{\frac{Jg^{ij}}{g_{33}}\frac{\p\Theta}{\p x^j}}-\frac{d\chi}{d\Theta}+\frac{1}{2g_{33}}\frac{dw_{3}^2}{d\Theta}+\frac{w_{3}}{J}\left[\frac{\p }{\p x^2}\lr{\frac{g_{13}}{g_{33}}}-\frac{\p }{\p x^1}\lr{\frac{g_{23}}{g_{33}}}\right]=0.\label{GGS1}
\end{equation}
Observe that
\begin{equation}
\begin{split}
\nabla\cdot&\lr{\frac{\p_3\cp\nabla x^3}{g_{33}}}\,dx\w dy \w dz=\mf{L}_{\frac{\p_3\cp\lr{\p_{1}\cp\p_{2}}}{g_{33}J}}dx\w dy \w dz\\
&=\mf{L}_{\frac{g_{32}\p_{1}-g_{31}\p_2}{g_{33}J}}dx\w dy \w dz=\mf{L}_{\frac{g_{32}\p_{1}-g_{31}\p_2}{g_{33}}}dx^1\w dx^2 \w dx^3\\
&=-\frac{1}{J}\left[\frac{\p }{\p x^2}\lr{\frac{g_{13}}{g_{33}}}-\frac{\p }{\p x^1}\lr{\frac{g_{23}}{g_{33}}}\right]\,dx\w dy \w dz.
\end{split}
\end{equation}
Finally, recalling that the Laplacian operator $\Delta=\p_{x}^2+\p_{y}^2+\p_{z}^2$ in the coordinates $\lr{x^1,x^2,x^3}$ is given by $\Delta=J^{-1}\p_{i}Jg^{ij}\p_{j}$, equation \eqref{GGS1}
can be rearranged as
\begin{equation}
\Delta\Theta-\nabla\Theta\cdot\nabla\log g_{33}-g_{33}\frac{d\chi}{d\Theta}+\frac{1}{2}\frac{dw_{3}^2}{d\Theta}-g_{33}w_{3}\nabla\cdot\lr{\frac{\p_{3}\cp\nabla x^3}{g_{33}}}=0,
\end{equation}
which is the Grad-Shafranov equation.

%\end{normalsize}

\end{document}